\documentclass[letterpaper,twocolumn,10pt,accepted=2024-11-07]{quantumarticle}
\pdfoutput=1

\usepackage{amsfonts}%
\usepackage{amsmath}%
\usepackage{amssymb}%
\usepackage{graphicx}
\usepackage{mathtools}
\usepackage[abs]{overpic}
\usepackage{hyperref}

\providecommand{\U}[1]{\protect\rule{.1in}{.1in}}
\newtheorem{theorem}{Theorem}

\newtheorem{corollary}{Corollary}

\newtheorem{definition}{Definition}

\newtheorem{lemma}{Lemma}

\newtheorem{proposition}{Proposition}
\newtheorem{remark}[theorem]{Remark}

\newenvironment{proof}[1][Proof]{\noindent\textbf{#1.} }{\ \rule{0.5em}{0.5em}}

\newcommand{\eqdef}{\coloneqq}
\newcommand{\ua}{\uparrow}
\newcommand{\mN}{\mathcal{N}}
\newcommand{\lr}{\rangle\!\langle}
\newcommand{\cptp}{{\rm CPTP}}

\newcommand{\mR}{\mathcal{R}}
\newcommand{\mF}{\mathcal{F}}
\newcommand{\mE}{\mathcal{E}}
\newcommand{\mD}{\mathcal{D}}

\newcommand{\la}{\langle}
\newcommand{\ra}{\rangle}

 \def\tA{\tilde{A}}
\def\tB{\tilde{B}}

  \def\cp{{\rm CP}}

\newcommand{\be}{\begin{equation}}
\newcommand{\ee}{\end{equation}}

\def\tA{\tilde{A}}

\newcommand{\tr}{{\rm Tr}}
\def\u{\mathbf{u}}
\def\H{\mathbf{H}}
\def\D{\mathbf{D}}

\newcommand{\md}{\mathfrak{D}}

\allowdisplaybreaks

\begin{document} 

\title[Inevitability of knowing less than nothing]{Inevitability of knowing less than nothing}
\author{Gilad Gour}
\affiliation{Faculty of Mathematics, Technion - Israel Institute of Technology, Haifa 3200003, Israel}
\email{giladgour@technion.ac.il}
\author{Mark M. Wilde}
\affiliation{School of Electrical and Computer Engineering, Cornell University, Ithaca, New
York 14850, USA}
\email{wilde@cornell.edu}
\author{S. Brandsen}
\affiliation{Department of Physics, Duke University, Durham, North Carolina 27708, USA}
\author{Isabelle Jianing Geng}
\affiliation{Department of Mathematics and Statistics, Institute for Quantum Science and
Technology, University of Calgary, Alberta, Canada T2N 1N4}

\begin{abstract}
  A colloquial interpretation of entropy is that it is the knowledge gained upon learning the outcome of a random experiment. Conditional entropy is then interpreted as the knowledge gained upon learning the outcome of one random experiment after learning the outcome of another, possibly statistically dependent, random experiment. In the classical world, entropy and conditional entropy take only non-negative values, consistent with the intuition that one has regarding the aforementioned interpretations. However, for certain entangled states, one obtains negative values when evaluating commonly accepted and information-theoretically justified formulas for the quantum conditional entropy, leading to the confounding conclusion that one can know less than nothing in the quantum world. Here, we introduce a physically motivated framework for defining quantum conditional entropy, based on two simple postulates inspired by the second law of thermodynamics (non-decrease of entropy) and extensivity of entropy, and we argue that all plausible definitions of quantum conditional entropy should respect these two postulates. We then prove that all plausible quantum conditional entropies take on negative values for certain entangled states, so that it is inevitable that one can know less than nothing in the quantum world. All of our arguments are based on constructions of physical processes that respect the first postulate, the one inspired by the second law of thermodynamics. 
\end{abstract}

\maketitle

\tableofcontents


\section{Introduction}

Quantum information science is famously known for having features that are not present in classical information science. Among these, perhaps the feature that distinguishes it the most from the classical case is that our knowledge of the state of two particles can be greater than our knowledge of the individual states of the particles. Indeed, for an Einstein--Podolsky--Rosen (EPR) state $|\Phi^{(2)}\rangle \coloneqq (|00\rangle + |11\rangle)/\sqrt{2}$~\cite{epr1935}, said to be entangled, our knowledge of the joint state of the two particles is complete. However, the state of an individual particle is described by an unbiased probabilistic ensemble  (i.e., $|0\rangle$ or $|1\rangle$ with equal probability), and we thus have incomplete knowledge of the state of an individual particle~\cite{S35}.

Given this peculiar situation, one of the earliest puzzles in quantum information science was to devise a quantum generalization of the concept of conditional entropy. In the classical case, conditional entropy is meant to capture the uncertainty about the state of one particle given access to a second particle that is potentially statistically dependent on the first~\cite{S48}. If there is no statistical dependence whatsoever, then the conditional entropy reduces to the usual entropy. However, if there is dependence, then there is a reduction in the uncertainty of the first particle when given access to the second. When generalizing the definition of conditional entropy in a straightforward way to the quantum case (by means of this reduction), one finds that quantum conditional entropy is negative when evaluated for the EPR state mentioned above~\cite{CA97}. Since conditional entropy is never negative in classical information science, this situation represents a radical departure from the classical case and has been popularly described as ``knowing less than nothing''~\cite{T05}.

In quantum information theory~\cite{H17,W17,Watrous2018,H19book,KW20book}, the traditional approach to resolving such puzzles is to devise a physically meaningful task for which an information quantity is an optimal rate for accomplishing the task. For the case of quantum conditional entropy, this approach was successfully applied and led to the introduction of the quantum state merging protocol~\cite{Horodecki2005,Horodecki:2007:107}. In this protocol, the goal is for one party $A$ to use entanglement and classical communication to transfer her share~$A$ of a joint state to another party $B$. If the conditional entropy $H(A|B)$ is non-negative, then the protocol consumes entanglement to accomplish the task, at a rate given by $H(A|B)$ EPR states per copy of the original state. Whereas, if the conditional entropy $H(A|B)$ is negative, then the protocol generates entanglement at a rate of $H(A|B)$ EPR states per copy of the original state, while transferring $A$ to~$B$. This protocol inspired further developments regarding conditional entropy in the context of the uncertainty principle~\cite{BCCRR10} and thermodynamics~\cite{Rio2011}, in which its negativity plays an essential role. 

Although these works brought great insight to the peculiar situation mentioned above, they mainly focus on a particular definition of the conditional entropy, based on the formula from~\cite{CA97}. It is also well known that there are an infinite number of ways to define quantum conditional entropy, with particular examples being the conditional min-entropy~\cite{Renner2005}, conditional max-entropy~\cite{KRS09}, conditional R\'enyi entropies~\cite{TCR09,MDSFT13,TBH13}, and those based more broadly on generalized divergences (see, e.g.,~\cite{KW20book}). Several of these alternate definitions have physical meanings as well~\cite{KRS09,HT16}.  However, all of these definitions are based on particular mathematical formulas for quantum conditional entropy that hitherto have not been derived on the basis of simple physical principles. As such, the aforementioned contributions do not address the question of why any plausible quantum conditional entropy can be negative when defining it in the most general way possible, which is consistent with some simple physical principles.

\subsection{Summary of results}

\label{sec:summary-results}

In this paper, it is our aim to address this pressing foundational question, and we do so by introducing a physically motivated framework for defining quantum conditional entropy, based on two simple postulates. The two postulates are related to the second law of thermodynamics and the extensivity of entropy, widely accepted in scenarios of physical interest. This novel framework is essentially the simplest way of defining any plausible quantum conditional entropy, and one of our main conclusions is that any such quantum conditional entropy must take on a negative value when evaluated for the EPR state. 
Following the popular terminology introduced in~\cite{T05}, we can summarize this finding  as the ``inevitability of knowing less than nothing'' in the quantum world. As such, our approach answers the aforementioned question, i.e., why the quantum conditional entropy can be negative. We additionally note here that this finding is in stark contrast to the conclusion, also presented here, that any plausible unconditional quantum entropy takes on a non-negative value for all quantum states (here quantum entropy is defined by a similar physically motivated framework related to that from~\cite{LY1999,LY1998} and that generalizes the classical case in~\cite{GT2021}).

One may wonder which of the two postulates is responsible for conditional entropy becoming negative in the quantum world. Interestingly, we identify that it is the extensivity postulate that does so. To arrive at this conclusion, we provide an example of a non-negative measure of conditional entropy that satisfies the monotonicity postulate of conditional entropy (i.e., the one related to the second law) and exhibits weak additivity for tensor-product states. This measure, essentially a regularized version of the ``measured'' conditional entropy, does not exhibit full additivity under arbitrary tensor-product states. This distinction is important, as it highlights that fully quantum additivity under tensor-product states commands negativity. See Appendix~\ref{appg} for further details.

As additional findings, we show that all plausible quantum conditional entropies cannot be smaller than the quantum conditional min-entropy, thus justifying once and for all the name of the latter quantity as the smallest plausible quantum conditional entropy. We also establish a logical equivalence between the non-negativity of all quantum conditional entropies and the well known reduction criterion~\cite{HH99} from entanglement theory. As a corollary, it follows that all separable states have non-negative quantum conditional entropy. 

In what follows, we present details of our claims. We first generalize the classical entropy framework from~\cite{GT2021} in order to define quantum entropy, and then we introduce our framework for defining quantum conditional entropy. After that, we indicate several properties of quantum conditional entropy that follow from the two postulates, before giving arguments for our main results. All arguments for our main results are based on constructions of physical processes that respect the first aforementioned postulate for quantum conditional entropy.  The only background needed to understand the rest of our paper is basic knowledge of density operators and quantum channels (physical processes), available in various textbooks on quantum information~\cite{H17,W17,Watrous2018,H19book,KW20book}.

\subsection{Notation}

Throughout our paper, we employ standard notation used in quantum information~\cite{H17,W17,Watrous2018,H19book,KW20book}. We use symbols like $\rho$, $\sigma$, $\omega$ to refer to quantum states (density operators) and $\mathfrak{D}(A)$ to denote the set of density operators acting on a Hilbert space $A$. We also identify a quantum system $A$ with its corresponding Hilbert space~$A$. Let $\u_A \coloneqq I_A / |A|$ denote the uniform state of a system~$A$, where $I_A$ is the identity operator acting on $A$ and $|A|$ denotes the dimension of~$A$. We use system labels to denote bipartite states as $\rho_{AB}, \sigma_{AB}, \omega_{AB} \in \mathfrak{D}(AB)$, and $\rho_A \equiv \operatorname{Tr}_B[\rho_{AB}]$ denotes the marginal state after performing a partial trace over system $B$. A quantum channel $\mathcal{N}_{A\to B}$ is a completely positive and trace-preserving map that takes an input system $A$ to an output system $B$. If the input and output systems are the same, then we simply write $\mathcal{N}_A$. A particular kind of quantum channel is an isometric channel, typically denoted by $\mathcal{U}_{A\to B}$ or $\mathcal{V}_{A\to B}$, and is realized as $\mathcal{U}_{A\to B}(\cdot) = U_{A\to B}(\cdot)(U_{A\to B})^\dag$, where $U_{A\to B}$ an isometry (i.e., it satisfies $(U_{A\to B})^\dag U_{A\to B} = I_A$).

\section{Uncertainty and entropy}

\subsection{Uncertainty measures}

Let us first define what it means for a function to be an uncertainty measure, following the approach of~\cite{GT2021} for the classical case. What we expect for any measure of uncertainty is that it does not decrease under the action of a mixing operation. Thus, we  demand that an uncertainty measure is a function that is monotone with respect to a set of mixing operations, similar to how one variant of the second law of thermodynamics states that the entropy of an isolated system left to spontaneous evolution does not decrease. This approach is also consistent with that taken in quantum resource theories when defining resource monotones~\cite{Chitambar_2019}, in turn inspired by the second law of thermodynamics. See also~\cite{HHH+2003,HHO2003}.

In the classical case of finite and discrete probability distributions, such a mixing operation is described by a random relabeling of values. For example, for a six-sided die described by a probability distribution, one can randomly relabel or permute the numbers on the faces of the die to realize a mixing operation, and the outcome of the die toss becomes more difficult to guess (and intuitively, more uncertain). Furthermore, the odds of winning a game of chance are lower after performing a mixing operation~\cite{BGG22}, providing a precise way of measuring uncertainty. 

Under such mixing operations,  the uniform distribution remains invariant. Furthermore, our intuition is that the uniform distribution should take the largest value of an uncertainty measure for all probability distributions of the same size, given that it is an unbiased distribution with equally likely outcomes and thus the unique distribution for which it is the most difficult to guess which outcome will occur (more generally, for which one is the least likely to win a game of chance~\cite{BGG22}). Given everything stated above, one thus defines a function to be an uncertainty measure in the classical case if it is not equal to the zero function (i.e., a function that trivially takes all inputs to the value zero) and if it does not decrease under the action of a mixing operation~\cite{GT2021}.

In order to generalize this notion to the quantum case, the key previously stated property that we use to define a mixing operation is the preservation of the uniform state~$\u$. In this way, it is guaranteed that a mixing operation does not increase the value of an uncertainty measure for the uniform state; rather, it remains constant. Indeed, such a state is already intuitively maximally uncertain, as argued above, and so it cannot be any more uncertain, among all states of a fixed dimension. At the very least, the mixing operations we consider for defining an uncertainty measure should be quantum channels because all physical processes are described by quantum channels. The additional property of preserving the uniform state thus constrains the mixing operations to be unital channels (i.e., channels that preserve the identity operator). Formally, recall that a channel $\mathcal{N}_A$ is unital if $\mathcal{N}_A(I_A)=I_A$, or equivalently, if $\mathcal{N}_A(\u_A)=\u_A$~\cite{H17,W17,Watrous2018,H19book,KW20book}. We thus define a function of quantum states to be an uncertainty measure if it is not equal to the zero function and if it does not decrease under the action of a unital channel~$\mathcal{N}_{A}$, i.e., 
\begin{equation}
f(\mathcal{N}_{A}(\rho_A)) \geq f(\rho_A)
\end{equation}
for every state $\rho_A$ and unital channel~$\mathcal{N}_{A}$. 

A state $\rho_A$ majorizes another state $\sigma_A$, denoted by $\rho_A \succ \sigma_A$, if the following inequalities hold for all $k\in \{1, \ldots, |A|\}$:
\begin{equation}
\sum_{i=1}^k \lambda_i^{\downarrow}(\rho_A) \geq \sum_{i=1}^k \lambda_i^{\downarrow}(\sigma_A),
\end{equation}
where $\lambda_i^{\downarrow}(\omega_A)$ denotes the components of the vector of eigenvalues of a density operator $\omega_A$ sorted in decreasing order.
It is well known that $\rho_A$ majorizes $\sigma_A$ if and only if there exists a unital channel $\mathcal{N}_A$ such that $\sigma_A = \mathcal{N}_A(\rho_A)$~\cite{U71,U72,U73} (see also~\cite[Theorem~4.32]{Watrous2018}). Thus, our definition of an uncertainty measure $f$ implies that $f$  reverses the majorization pre-order, i.e.,
\begin{equation}
\rho_A \succ \sigma_A \quad \Rightarrow \quad f(\rho_A) \leq f(\sigma_A).
\end{equation}

It is ideal and natural to generalize the definition of an uncertainty measure just slightly in order to compare the uncertainties of states of different dimensions. Suppose that we have a state $\rho_A$ with an uncertainty value given by $f(\rho_A)$ for some uncertainty measure $f$. Now consider the density matrix $\rho_{A'} \coloneqq \rho_A \oplus {\bf 0}$, where ${\bf 0}$ denotes the zero square matrix, so that $\rho_{A'}$ represents an embedding of $\rho_A$ into a larger system $A'$. There is no reason that the uncertainty measure should change at all under this embedding, because, intuitively, the extra matrix entries corresponding to ${\bf 0}$ neither increase nor decrease uncertainty. Indeed, when flipping a die on a table, one could add an extra degree of freedom corresponding to whether the table is upside down or right side up, and under normal circumstances, the table will never be upside down, so that the uncertainty of the die toss will be unaffected by this embedding. As such, we impose the additional constraint that the uncertainty measure $f$ satisfies the equality $f(\rho_A) = f(\rho_{A'})$ for every state $\rho_A$ and every embedding into a larger system~$A'$.

In this spirit, we now generalize the standard definition of majorization for quantum states from~\cite{U71,U72,U73} (see also~\cite{Watrous2018}).
\begin{definition}[Majorization]
Given two states $\rho_{A}$ and $\sigma _{A'}$, we say that $\rho_{A}$ majorizes $\sigma _{A'}$, and write 
\begin{equation}
    \rho_{A}\succ\sigma_{A'}
    \label{eq:final-maj-def}
\end{equation}
if there exist a system $A''$ such that $|A''| \geq \max\{|A|,|A'|\}$, isometric channels $\mathcal{V}_{A\to A''}$ and $\mathcal{U}_{A'\to A''}$,  and a unital channel $\mathcal{N}_{A''}$ such that
\begin{equation}
    \mathcal{U}_{A'\to A''}(\sigma _{A'}) =(\mathcal{N}_{A''}\circ \mathcal{V}_{A\to A''})(\rho_{A}) \;.
    \label{eq:iso-conditions-majorization}
\end{equation}
\end{definition}

The isometric channels introduced in the definition of majorization above enable us to compare states of different dimensions, as desired. Also, note that majorization between two probability vectors of the same dimension can be extended to vectors of different dimensions in this way (it is a standard approach to add zero components to the vector with the smaller dimension so that the resulting vectors have equal dimension). As discussed in the previous paragraph, such an extension is motivated by the fact that uncertainty measures should not change under embeddings.

In conclusion, our final definition of an uncertainty measure  is a function $f$ that is not equal to the zero function and reverses the majorization pre-order from~\eqref{eq:final-maj-def}; i.e., for all states $\rho_A$ and $\sigma_{A'}$,
\begin{equation}
\rho_A \succ \sigma_{A'} \quad \Rightarrow \quad f(\rho_A) \leq f(\sigma_{A'}).
\end{equation}

\subsection{Entropy}

Following the approach of~\cite{GT2021} for the classical case, an uncertainty measure is not necessarily an entropy. To be an entropy, we also require that the uncertainty measure possess an extensivity or additivity property. That is, $f$ is an entropy if it is an uncertainty measure and it is additive on tensor-product states; i.e., for all states $\rho_A$ and $\sigma_{A'}$,
\begin{equation}
f(\rho_A \otimes \sigma_{A'}) = f(\rho_A) + f(\sigma_{A'}).
\end{equation}
The additivity postulate of entropy is motivated by the basic expectation that the disorder of isolated, non-interacting systems should be the sum of the individual disorders. Indeed, if two physical systems are separate from each other, not having interacted in any way whatsoever previously, we expect that the disorder characterizing the overall system should simply be the sum of the individual disorders (i.e., there would not be a meaningful alternative way of quantifying the overall disorder in this scenario).

The motivation for additivity of entropy is also related to the second law of thermodynamics, in particular, the Clausius and Kelvin--Planck statements. In those statements, cyclic processes are considered, in which a system of interest undergoes a thermodynamic transition, while all other systems, including the environment, heat baths, etc., all return to their original state. In recent developments, using a quantum information-theoretic approach to small-scale thermodynamics~\cite{BHNOW15}, these cyclic processes were identified as catalytic processes. That is, consider a thermodynamical evolution of a physical system $A$, in a state~$\rho_A$, to a physical system $A'$ in a state~$\sigma_{A'}$. The thermal machine that corresponds to this transition includes all the other systems (e.g., heat baths, environment, etc.) that can be represented with an additional system $M$ in a state~$\tau_M$. Therefore, for cyclic processes, the thermodynamical transition
takes the form
\begin{equation}
\label{transition}
\rho_{A}\otimes\tau_{M}\to\sigma_{A'}\otimes\tau_{M}\;.
\end{equation}
In this context, the second law not only states that the entropy of system $A$ is no greater than that of system~$A'$, 
but also that this property holds if and only if the entropy of the joint state $\rho_{A}\otimes\tau_M$  increases (or remains unchanged) in such a thermodynamic cyclic process that returns $\tau_{M}$ intact. Clearly, if the entropy is measured with an additive function (under tensor products) then the entropy of  $\rho_A$ is never greater than that of $\sigma_{A'}$ if and only if the same relation holds between $\rho_{A}\otimes\tau_{M}$ and $\sigma_{A'}\otimes\tau_{M}$. Thus, the second law of thermodynamics also motivates the additivity postulate for entropy.

With all of this reasoning in place, we now provide a formal definition of entropy:
\begin{definition}[Entropy]
    A function
\begin{equation}
\H : \bigcup_{A} \mathfrak{D}(A) \rightarrow \mathbb{R}
\label{eq:entropy-def-1}
\end{equation}
is a quantum entropy if it is not equal to the zero function and satisfies the following two postulates for all systems~$A$ and $A'$ and states $\rho_{A}$ and $\sigma_{A'}$:
\begin{enumerate}
    \item Monotonicity:  
    \begin{equation}
\rho_{A} \succ \sigma_{A'} \quad \Rightarrow \quad   \H(\rho_A) \leq \H(\sigma_{A'}).    
    \label{eq:entropy-def-2}
    \end{equation}
    
    \item Additivity:
    \begin{equation}
    \H(\rho_A \otimes \sigma_{A'}) = \H(\rho_A) + \H(\sigma_{A'})    .
    \label{eq:entropy-def-3}
    \end{equation}
\end{enumerate}
\end{definition}

As stated previously, the requirement that an entropy not be equal to the zero function is to avoid a trivial situation in which the entropy of every state takes on the value zero.


By employing the two basic postulates of entropy and the same reasoning given in~\cite[Lemma~3]{GT2021}, some important properties follow (see Appendix~\ref{app:entropy-props} for proofs). The entropy of every state~$\rho$ is non-negative, i.e., $\H(\rho) \geq 0$, and it is equal to zero for every pure state $|\psi\rangle \! \langle \psi |$, i.e., $\H(|\psi\rangle \! \langle \psi |) = 0$. For every state $\sigma$ of a  fixed dimension~$d$, the entropy does not exceed the entropy of the uniform state of dimension $d$, i.e., $\H(\u^{(d)}) \geq \H(\sigma)$. Finally, the entropy of the uniform state of dimension two is strictly positive, i.e., $\H(\u^{(2)}) > 0$. This allows for setting the normalization of entropy to be in units of bits, such that $\H(\u^{(2)})=1$, and henceforth, we thus set the normalization of entropy in this way. It then follows that $\H(\u^{(d)}) = \log_2 d$.

\section{Conditional uncertainty and entropy}

\subsection{Conditional uncertainty measures}

\label{sec:q-cond-maj}

Our goal now is to define what it means for a function of a bipartite state~$\rho_{AB}$ to be a measure of conditional uncertainty. 
Inspired by the criterion of monotonicity under mixing operations in the previous case of unconditional uncertainty, we are interested in identifying a criterion for a bipartite channel $\mathcal{N}_{AB \to A'B'}$ to be a conditionally mixing operation. For the unconditional case, we ultimately stated that a quantum channel is a mixing operation if it preserves the uniform state, and this choice led to identifying unital channels as the mixing operations that define an uncertainty measure. The reason for this, as previously stated, is that the uniform state is the least informative state, it thus gives the least likelihood of winning a game of chance, and should thus be a state of maximal uncertainty (for which a mixing operation cannot increase uncertainty).

For the setting of conditional uncertainty, it remains the case that the uniform state $\u_{AB} = \u_{A} \otimes \u_B$ is the least informative state, for reasons similar to those given before. However, for conditional uncertainty, we are interested in the uncertainty of system  $A$ in the presence of system $B$. That is, the system $B$ can be used as an aid or side information to help win a game of chance performed on the system $A$. With this in mind, a state of the form $\u_A \otimes \rho_B$ is still least informative, with $\rho_B$ an arbitrary state, because the system $B$ is independent of system $A$ and thus cannot help at all in winning a game of chance performed on $A$. At the same time, the state of system $A$ is still uniform and thus the least informative. As such, we identify the set of states having the form $\u_A \otimes \rho_B$, where $\rho_B$ is an arbitrary state of system~$B$, as being the set of states with maximum conditional uncertainty.

Thus, we have argued that one requirement for a function to be a measure of conditional uncertainty is that it should not decrease under the action of a bipartite channel that preserves the set of states with maximum conditional uncertainty. We call such channels conditionally unital channels because they generalize unital channels.
Formally,
a bipartite channel $\mathcal{N}_{AB \rightarrow {AB'}}$ is a conditionally unital channel if for every $\rho_B \in \mathfrak{D}(B)$, there exists a state $\sigma_{B'} \in \mathfrak{D}(B')$ such that
\begin{align}
    \mathcal{N}_{AB \rightarrow {AB'}}(\mathbf{u}_{A} \otimes \rho_{B}) = \mathbf{u}_{A} \otimes \sigma_{B'}.
    \label{eq:cond-unital-def}
\end{align}
See Appendix~\ref{app:Choi-cond-unital} for an equivalent formulation of this condition in terms of the Choi matrix of $\mathcal{N}_{AB \rightarrow {AB'}}$. We note here that conditionally unital channels were independently defined in our companion paper~\cite{BGWG21} and in~\cite{Vempati2022unitaloperations}. See also~\cite[Lemma~3.1.12]{Renner2005} for an implicit definition of sub-conditionally-unital channels.




In the setting of conditional uncertainty, it is important to consider a further constraint on the class of transformations allowed. Intuitively, conditional uncertainty measures the uncertainty of system $A$ when one is given access to system $B$. Evidently, by allowing transformations in which information can leak from system $A$ to system~$B$, Bob could gain information that could decrease the uncertainty associated with system~$A$. Thus, it is natural to assume that channels that do not decrease conditional uncertainty do not allow for communication or signaling from Alice to Bob.  This requirement is also motivated from a cryptographic perspective, in which conditional uncertainty measures the amount of uncertainty that an eavesdropper has about another system that is intended to be secure~\cite{PR22} (with respect to our system labeling, the eavesdropper would have system $B$ and the secure system would be~$A$).
With this motivation in mind, let us recall the definition of $A\not \to B$ semi-causal channels~\cite{Beckman_2001}.  
A bipartite channel $\mathcal{N}_{AB \rightarrow AB'}$ is $A\not\to B'$ semi-causal if for every channel $\mathcal{M}_{A}$
\begin{align}
    \mathcal{N} _{AB \rightarrow B'} \circ \mathcal{M} _{A } = \mathcal{N} _{AB \rightarrow B'},
\end{align}
where
\begin{equation}
\mathcal{N} _{AB \rightarrow B'} \coloneqq  \operatorname{Tr}_A \circ \mathcal{N} _{AB \rightarrow AB'}.
\end{equation}
See Figure~\ref{fig_semi-causal} for a depiction of the semi-causality requirement. 

\begin{figure}
  \begin{overpic}[scale=0.5]{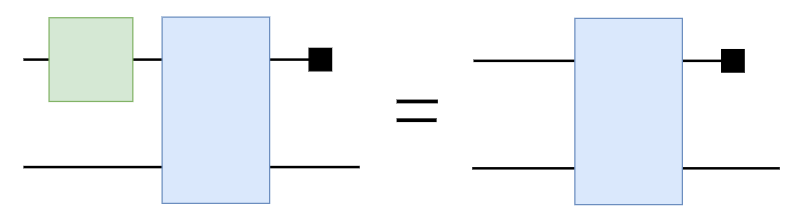}
  \put(20, 44){$\mathcal{M}$}
\put(59, 30){$\mathcal{N}$}
\put(183, 30){$\mathcal{N}$}
\end{overpic}
\caption{Depiction of the definition of a semi-causal channel~$\mathcal{N}$, where the black square represents the discarding channel.}
\label{fig_semi-causal}
\end{figure}

The authors of~\cite{ESW02} proved that every semi-causal channel can be written as a local channel on $B$ that feeds an output (corresponding to a quantum system~$R$) into a channel acting only on $A$ and $R$. See also~\cite{Piani_2006} and Figure~\ref{fig_semi-casual_channel} for a visual depiction. As such, for these channels, information can only flow from $B$ to $A$, and it never leaks from~$A$ to~$B$. Formally, if $\mathcal{N} _{AB \rightarrow AB'}$ is an $A \not\to B'$ semi-causal channel, then there exists a quantum system~$R$, a quantum channel $\mathcal{E}_{AR \rightarrow A}$, and an isometric channel~$\mathcal{F}_{B \rightarrow RB'}$ such that
\begin{equation}
    \mathcal{N} _{AB \rightarrow AB'} = 
 \mathcal{E}_{RA \rightarrow A} \circ \mathcal{F}_{B \rightarrow RB'} \;.
 \label{eq:semi-causal-breakdown}
\end{equation}

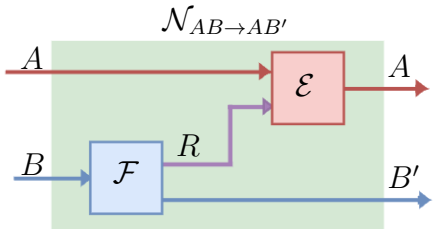
\begin{figure}
\large
  \begin{overpic}[scale=0.6]{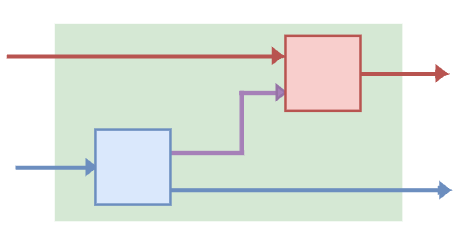}
  \put(62, 80){$\mathcal{N}_{AB\to AB'}$}
\put(43, 22){$\mathcal{F}$}
\put(112, 55){$\mathcal{E}$}
\put(67,32){$R$}
\put(8,25){$B$}
\put(8,65){$A$}
\put(147,62){$A$}
\put(147,20){$B'$}
\end{overpic}
\caption{Eq.~\eqref{eq:semi-causal-breakdown} asserts that a semi-causal bipartite channel $\mathcal{N}_{AB \rightarrow AB'}$ can be implemented in the depicted fashion.}
\label{fig_semi-casual_channel}
\end{figure}


We use the term conditionally mixing to describe a channel $\mathcal{N}_{AB\to AB'}$ that is both conditionally unital and $A \not\to B'$ semi-causal. Note that the properties of such a channel are invariant under the exchange of the $A$ input and $A$ output systems. This is most readily observed using the formalism of Choi matrices and we discuss it further in Appendix~\ref{app:locally-balanced-choi}. Importantly, we use conditionally mixing channels to define measures of conditional uncertainty.

%

Generalizing how unital channels define the majorization pre-order on quantum states, we now define a conditional majorization pre-order on bipartite quantum states, based on conditionally mixing channels. 
\begin{definition}[Conditional majorization]
\label{def:cond-maj-cond-mix}
Given two bipartite states $\rho_{AB}$ and $\sigma _{A'B'}$, we say that $\rho_{AB}$ conditionally majorizes $\sigma _{A'B'}$  with respect to system $A$, and write 
\begin{equation}
    \rho_{AB}\succ _{ A}\sigma_{A'B'}
        \label{eq:cond-maj-def-1}
\end{equation}
if there exist a system $A''$ such that $|A''| \geq \max\{|A|,|A'|\}$, isometric channels $\mathcal{V}_{A\to A''}$ and $\mathcal{U}_{A'\to A''}$,  and a conditionally mixing channel~$\mathcal{N}_{A''B \rightarrow A''B'}$  such that
\begin{equation}
    \mathcal{U}_{A'\to A''}(\sigma _{A'B'}) =(\mathcal{N}_{A''B\to A''B'} \circ \mathcal{V}_{A\to A''})(\rho_{AB}) \;.
        \label{eq:cond-maj-def-3}
 \end{equation}
%
\end{definition}

The isometric channels introduced in the definition above enable us to compare bipartite states of different dimensions, as was the case before with the generalized definition of majorization in~\eqref{eq:final-maj-def}--\eqref{eq:iso-conditions-majorization}. We note here that conditional majorization, as defined here, reduces to classical conditional majorization, as defined in~\cite{GGHKLV18}, when the the states involved are classical (see Appendix~\ref{app:q-cond-maj-reduces-to-classical-cond-maj} for a proof). Specifically,
we show that in the classical domain, conditionally mixing channels coincide with the conditionally doubly stochastic (CDS) channels that were defined in~\cite{GGHKLV18}.

We define a function $f$ to be a conditional uncertainty measure if it is not equal to the zero function and if it reverses the  conditional majorization pre-order, i.e., 
\begin{equation}
\rho_{AB} \succ_A \sigma_{A'B'} \quad \Rightarrow \quad f(\rho_{AB}) \leq f(\sigma_{A'B'}).
\end{equation}

\subsection{Conditional entropy}

Having defined measures of conditional uncertainty, we now follow the approach from before and define a function to be a conditional entropy if it is not only a conditional uncertainty measure but also additive on tensor-product states. Our motivation for the additivity postulate is the same as before.

\begin{definition}[Conditional entropy]
    A function
\begin{equation}
\H : \bigcup_{A, B} \mathfrak{D}(AB) \rightarrow \mathbb{R}
\label{eq:cond-ent-gen-def-1st}
\end{equation}
is a quantum conditional entropy if it is not equal to the zero function when $|B|=1 $ and satisfies the following postulates for all systems $A$, $B$, $A'$, $B'$ and states $\rho_{AB}$ and $\sigma_{A'B'}$:
\begin{enumerate}
    \item Monotonicity: 
    \begin{equation}
      \rho_{AB}\succ _{ A} \sigma_{A'B'} \quad \Rightarrow \quad \H(A|B)_{\rho} \leq \H(A'|B')_{\sigma}. 
    \label{eq:QCE-1st-postulate-mono}   
    \end{equation}
    
    \item Additivity:
    \begin{equation}
    \H(AA'|BB')_{\rho \otimes \sigma} = \H(A|B)_{\rho} + \H(A'|B')_{\sigma}    .
        \label{eq:QCE-2nd-postulate-add}
    \end{equation}

\end{enumerate}
\end{definition}

%

As indicated in the definition, a conditional entropy~$\H$ is defined for all density matrices in all finite dimensions. As a very special case, if $\rho\in\mathfrak{D}(AB)$ and $B$ is a trivial system for which $|B|=1$, then we write
\begin{equation}
    \H(A|B)_\rho=\H(A)_\rho\;,
\end{equation}
where $\H(A)_\rho$ is an entropy, as defined before.


Here we state some basic properties of quantum conditional entropy before presenting our main result in the next section (see Appendix~\ref{app:cond-ent-proofs} for proofs).
First, conditional entropy reduces to entropy when evaluated on a tensor-product state. That is, for every product state $\omega_{A} \otimes \tau_{B}$, the following equality holds: 
\begin{equation}
    \H(A|B)_{ \omega\otimes\tau}= \H(A)_{\omega}.
\end{equation}
Next, conditional entropy is invariant under the action of local isometric channels. That is, for a bipartite state~$\rho_{AB}$ and isometric channels~$\mathcal{U}_{A\to A'}$ and $\mathcal{V}_{B\to B'}$,
\begin{equation}
    \H(A|B)_{\rho} = \H(A'|B')_{\omega},
\end{equation}
where $\omega_{AB'} \coloneqq (\mathcal{U}_{A\to A'} \otimes \mathcal{V}_{B\to B'})(\rho_{AB})$.
Finally, 
consider the product state $\u^{(2)}_A \otimes \rho_B$. Similar to what we found previously for unconditional entropy, the following inequality holds
\begin{equation}
\H(A|B)_{\u^{(2)} \otimes \rho}  = \H(A)_{\u^{(2)}}> 0.
\end{equation}
This strict inequality allows us to set a normalization factor for conditional entropy. To be consistent with the normalization convention for unconditional entropy, we set $\H(A|B)_{\u^{(2)} \otimes \rho} = 1$, which in turn implies that
\begin{equation}
\H(A|B)_{\u^{(d)} \otimes \rho} = \log_2 d,
\end{equation}
where the state being evaluated is $\u^{(d)}_A \otimes \rho_B$.

\section{Main result: Inevitability of negative quantum conditional entropy}

We now have almost everything set up to state our main result. Before doing so, let us define  the conditional min-entropy~\cite{Renner2005} of a bipartite state $\rho_{AB}$ as
\begin{equation}
H_{\min}(A|B)_{\rho}  \coloneqq - \inf_{\lambda \geq 0} \log_2 \{\lambda : \rho_{AB} \leq \lambda I_A \otimes \rho_B\}
\label{eq:cond-min-ent-def}, 
\end{equation}
The conditional min-entropy is indeed  a conditional entropy according to~\eqref{eq:QCE-1st-postulate-mono}--\eqref{eq:QCE-2nd-postulate-add} (see Appendix~\ref{app:cond-min-ent-is-cond-ent} for a proof). 
This quantity was previously given the name conditional min-entropy because it is known to be the least among all R\'enyi conditional entropies~\cite{Renner2005,TCR09,MDSFT13}. As part of our main result, we strengthen this observation by proving that all plausible quantum conditional entropies are not smaller than the conditional min-entropy; that is, the conditional min-entropy is the smallest plausible conditional entropy.

The following theorem (proved in Appendix~\ref{app:proof-main-thm-main-txt}) states that every plausible quantum conditional entropy is bounded from below by the conditional min-entropy and also that this lower bound is saturated for the maximally entangled state. As a consequence, every plausible conditional entropy takes on a negative value for the maximally entangled state.

\begin{theorem}
\label{thm:neg-cond-ent}
Let $\H$ be a quantum conditional entropy, as defined in~\eqref{eq:cond-ent-gen-def-1st}--\eqref{eq:QCE-2nd-postulate-add}. Then, for every state $\rho_{AB}$,
\begin{equation}
\H(A|B)_\rho\geq H_{\min}(A|B)_\rho.
\label{eq:main-thm}
\end{equation}
Equality is attained in~\eqref{eq:main-thm} if $\rho_{AB}$ is equal to a maximally entangled state $\Phi^{(k)}_{A'B'}$ by the action of local isometric channels. Thus,
\begin{equation}
\H(A|B)_{\Phi^{(k)}} = -\log_2 k
\end{equation}
for every conditional entropy $\H$.
\end{theorem}

The following statement is a direct consequence of Theorem~\ref{thm:neg-cond-ent}, indicating  that every plausible conditional entropy being non-negative is equivalent to the reduction criterion~\cite{HH99}, well known in entanglement theory. Thus, Corollary~\ref{cor:pos-QCE-reduct} provides a direct link between every conditional entropy (including the conditional min-entropy) and the reduction criterion.

\begin{corollary}
\label{cor:pos-QCE-reduct}
Let $\rho_{AB}$ be a bipartite state. Then the following are equivalent:
\begin{enumerate}
\item For every conditional entropy $\H$,
\begin{equation}
\H(A|B)_\rho\geq 0.
\end{equation}
\item The reduction criterion holds: $I_A\otimes\rho_B\geq\rho_{AB}$.
\end{enumerate}
\end{corollary}

\begin{proof}
The proof of Corollary~\ref{cor:pos-QCE-reduct} is immediate from Theorem~\ref{thm:neg-cond-ent}. Indeed, if the first condition holds, then the inequality $H_{\min}(A|B)_\rho\geq 0$ holds in particular (since $H_{\min}(A|B)_\rho$ is a conditional entropy), and the second condition is then a consequence of the definition in~\eqref{eq:cond-min-ent-def}. Now suppose that the second condition holds. This then means that there exists $\lambda \leq 1$  such that $\rho_{AB}\leq \lambda I_A\otimes\rho_B$. According to the definition in~\eqref{eq:cond-min-ent-def}, it follows that $H_{\min}(A|B)_\rho\geq 0$ and then Theorem~\ref{thm:neg-cond-ent} implies the first condition.
\end{proof}
\medskip

Suppose that a state $\rho_{AB}$ is separable~\cite{W89}, meaning that it can be written as a probabilistic mixture of product states (i.e., $\rho_{AB} = \sum_x p(x) \sigma^x_A \otimes \tau^x_B$ for some probability distribution $\{p(x)\}_x$ and sets $\{\sigma^x_A\}_x$ and $\{\tau^x_B\}_x$ of states). Then it follows that the second condition in Corollary~\ref{cor:pos-QCE-reduct} holds~\cite{HH99}, so that every conditional entropy is non-negative for all separable states. 

Although all plausible conditional entropies are non-negative on all separable states, this does not mean that the conditional entropies are only non-negative on separable states. In fact, some entangled states have non-negative conditional entropy on \emph{all} choices of conditional entropy functions. In light of this, Corollary~\ref{cor:pos-QCE-reduct} is helpful because it provides a simple sufficient criterion for a bipartite state to have non-negative quantum conditional entropy.

In Appendix~\ref{app:alt-cond-ent}, we discuss a subclass of conditional entropies and report results that are similar to those given in Theorem~\ref{thm:neg-cond-ent} and Corollary~\ref{cor:pos-QCE-reduct}. The difference is that the lower bound in~\eqref{eq:main-thm} is expressed in terms of a different definition of conditional min-entropy and the equivalent condition in Corollary~\ref{cor:pos-QCE-reduct} becomes a variation of the reduction criterion, rather than the reduction criterion itself.

In addition to the inequality~\eqref{eq:main-thm}, one can obtain other entropy inequalities that are based on the formalism developed here. For example, in Appendix~\ref{appG} we show that for every bipartite state $\rho_{AB}$,
\be
\H(AB)_\rho\geq \H(A|B)_{\rho}-\log_2|B|\;.
\ee
This inequality follows from the fact that the teleportation protocol of communicating a quantum system from Bob to Alice is a conditionally mixing channel (see Appendix~\ref{appG} for more details and other entropy inequalities).

\section{Conditional entropy from relative entropy}

In this section, we prove that one can construct a quantum conditional entropy from a quantum relative entropy, based on two simple postulates that define a quantum relative entropy~\cite{GT2021,GT20}. Before giving this construction, let us first recall the definition of a quantum relative entropy~\cite{GT2021,GT20}. 

\begin{definition}[Relative entropy]
A function
\begin{equation}
\D : \bigcup_{A}\{ \mathfrak{D}(A)\times \mathfrak{D}(A) \} \rightarrow \mathbb{R}_+ \cup \{\infty\}
\label{eq:rel-ent-gen-def-1st}
\end{equation}
is a quantum relative entropy if it is not equal to the zero function and satisfies the following postulates for all systems $A$ and $B$, all states $\rho_{A}$, $\sigma_{A}$, $\omega_B$, $\tau_B$, and every channel $\mathcal{N}_{A\to B}$:
\begin{enumerate}
    \item Monotonicity: 
    \begin{equation}
       \D(\rho_A \Vert \sigma_A) \geq \D(\mathcal{N}_{A\to B}(\rho_A) \Vert \mathcal{N}_{A\to B}(\sigma_A)). 
    \label{eq:rel-1st-postulate-mono}   
    \end{equation}
    
    \item Additivity:
    \begin{equation}
    \D(\rho_A \otimes \omega_B \Vert \sigma_A \otimes \tau_B)    = \D(\rho_A \Vert \sigma_A) + \D(\omega_B \Vert \tau_B).
        \label{eq:rel-2nd-postulate-add}
    \end{equation}
    
    \end{enumerate}
\end{definition}

Now we state our final result (see Appendix~\ref{app:proof-cond-ent-from-rel-ent} for a proof):

\begin{theorem}
\label{thm:cond-ent-from-rel-ent}
Let $\D$ be a quantum relative entropy. Then the function 
\begin{equation}
      \mathrm{log}_2 | A | - \D(\rho_{AB} \| \mathbf{u}_{A} \otimes \rho_{B})
    \label{eq:cond-ent-from-rel-ent}
\end{equation}
is a quantum conditional entropy.
\end{theorem}

Suitable relative entropies include the Umegaki (standard)~\cite{U62}, Belavkin--Staszewski~\cite{Belavkin1982}, Petz--R\'enyi~\cite{P85,P86}, sandwiched R\'enyi~\cite{MDSFT13,WWY14}, $\alpha$-$z$~\cite{AD15}, and geometric R\'enyi~\cite{M13,Matsumoto2018} relative entropies, for which we point to~\cite{KW20book} for a summary. The construction in~\eqref{eq:cond-ent-from-rel-ent}, for generating a conditional entropy from a relative entropy, has been known for some time, but the difference with our result here is that we prove, starting from the simple postulates for relative entropy in~\eqref{eq:rel-ent-gen-def-1st}--\eqref{eq:rel-2nd-postulate-add}, that the function in~\eqref{eq:cond-ent-from-rel-ent} satisfies the two simple postulates for conditional entropy in~\eqref{eq:cond-ent-gen-def-1st}--\eqref{eq:QCE-2nd-postulate-add}.

\section{Conclusion}

In summary, we have established a physically motivated, minimal framework for defining conditional entropies, based on two simple postulates (non-decrease under conditionally mixing channels and additivity on tensor-product states). We then explored consequences of these postulates, which include reduction to entropy for tensor-product states, invariance under local isometric channels, a lower bound in terms of the conditional min-entropy, and taking on a negative value for the maximally entangled state. We consider the latter property to be the most profound result of this work, and all of our arguments rely on constructing physical processes that respect the first postulate of non-decrease under conditionally mixing channels.

Returning to the point raised in the second paragraph of Section~\ref{sec:summary-results}, one may wonder if it might be possible to replace the additivity postulate of conditional entropy with a weaker criterion. For example, suppose we replace the additivity postulate~\eqref{eq:QCE-2nd-postulate-add} with a weaker postulate of the form
\begin{equation}\label{weakadditivity}
    \H(A^n|B^n)_{\rho^{\otimes n}} = n\H(A|B)_{\rho}\quad\forall n\in\mathbb{N}.
    \end{equation}
In Appendix~\ref{appg} we show that there exist examples of functions that satisfy both the monotonicity postulate and the above weak additivity, but take a non-negative value on \emph{all} states. Therefore, the negativity of conditional-entropies (as defined in this paper) on the maximally entangled state is closely related to the (strong) additivity postulate~\eqref{eq:QCE-2nd-postulate-add}.   

Every relative entropy $\D$ can be used as in~\eqref{eq:cond-ent-from-rel-ent} to define a conditional entropy.
One can also attempt to define another type of conditional entropy of the form
\be\label{36}
\log_2|A|-\min_{\sigma\in\md(B)}\D(\rho_{AB}\|\u_A\otimes\sigma_B)\;.
\ee
While the function above behaves monotonically under conditionally unital channels, and therefore satisfies the monotonicity postulate~\eqref{eq:QCE-1st-postulate-mono} it is not clear if this function is additive under tensor products. It is known to be additive if $\D$ is taken to be the Petz R\'enyi relative entropy~\cite[Lemma~3]{SW13} or the sandwiched R\'enyi relative entropy~\cite{Bei13}, but it is not known if it is additive in general.

From its definition, the set of all conditional entropies is convex. Hence, not every conditional entropy has to take the form of a difference of logarithmic dimension and divergence. For example, if $\D$ is one of the quantum R\'enyi relative entropies, then every convex combination of the functions~\eqref{eq:cond-ent-from-rel-ent} and~\eqref{36} is also a conditional entropy. 
It is an interesting open problem if there exists conditional entropies that cannot be expressed in terms of divergences.
Another interesting open problem is whether or not all quantum conditional entropies are upper bounded by the conditional max-entropy. We leave these open problems to future investigations.

\bigskip 

\textbf{Acknowledgments}.
The authors thank Nilanjana Datta, Kaiyuan Ji, and Henry Pfister for helpful discussions and Jonathan Sorce for feedback on our manuscript. GG and IG acknowledge support from the Natural Sciences and Engineering Research Council of Canada (NSERC). SB acknowledges support from the National Science Foundation (NSF) under Grant Nos.~1908730 and 1910571. MMW acknowledges support from the NSF under Grant No.~1907615. Any opinions, findings, conclusions, and recommendations expressed in this material are those of the authors and do not necessarily reflect the views of these sponsors.

\bigskip

\textbf{Author Contributions}.
The following describes the different contributions of all authors of this work, using roles defined by the CRediT
(Contributor Roles Taxonomy) project~\cite{NISO}:

\medskip 

\noindent \textbf{GG}: Conceptualization, Formal analysis, Investigation, Methodology, Project administration, Supervision, Validation, Writing - original draft, Writing - review \& editing.

\medskip 
\noindent \textbf{MMW}: Formal analysis, Funding acquisition, Investigation, Methodology, Validation, Writing - original draft, Writing - review \& editing.

\medskip 
\noindent \textbf{SB}: Conceptualization, Visualization

\medskip 
\noindent \textbf{IJG}: Conceptualization

\bibliography{ref}

\begin{thebibliography}{10}

\bibitem{epr1935}
Albert Einstein, Boris Podolsky, and Nathan Rosen.
\newblock ``Can quantum-mechanical description of physical reality be considered complete?''.
\newblock \href{https://dx.doi.org/10.1103/PhysRev.47.777}{Physical Review {\bf 47}, 777--780}~(1935).

\bibitem{S35}
Erwin Schr\"{o}dinger.
\newblock ``Discussion of probability relations between separated systems''.
\newblock \href{https://dx.doi.org/10.1017/S0305004100013554}{Proceedings of the Cambridge Philosophical Society {\bf 31}, 555--563}~(1935).

\bibitem{S48}
Claude~E. Shannon.
\newblock ``A mathematical theory of communication''.
\newblock \href{https://dx.doi.org/10.1002/j.1538-7305.1948.tb01338.x}{The Bell System Technical Journal {\bf 27}, 379--423}~(1948).

\bibitem{CA97}
Nicolas~J. Cerf and Christoph Adami.
\newblock ``Negative entropy and information in quantum mechanics''.
\newblock \href{https://dx.doi.org/10.1103/PhysRevLett.79.5194}{Physical Review Letters {\bf 79}, 5194--5197}~(1997).

\bibitem{T05}
Ian Turner.
\newblock ``Scientist knows less than nothing''~(2005).
\newblock Screenshot available at \url{https://twitter.com/markwilde/status/1558125303569846278}.

\bibitem{H17}
Masahito Hayashi.
\newblock ``Quantum information theory: Mathematical foundation''.
\newblock \href{https://dx.doi.org/10.1007/978-3-662-49725-8}{Springer}. ~(2017).
\newblock Second edition.

\bibitem{W17}
Mark~M. Wilde.
\newblock ``Quantum information theory''.
\newblock \href{https://dx.doi.org/10.1017/9781316809976}{Cambridge University Press}. ~(2017).
\newblock Second edition.

\bibitem{Watrous2018}
John Watrous.
\newblock ``The theory of quantum information''.
\newblock \href{https://dx.doi.org/10.1017/9781316848142}{Cambridge University Press}. ~(2018).

\bibitem{H19book}
Alexander~S. Holevo.
\newblock ``Quantum systems, channels, information: A mathematical introduction''.
\newblock \href{https://dx.doi.org/10.1515/9783110642490}{Walter de Gruyter}. ~(2019).
\newblock Second edition.

\bibitem{KW20book}
Sumeet Khatri and Mark~M. Wilde.
\newblock ``Principles of quantum communication theory: A modern approach''~(2024).
\newblock  \href{http://arxiv.org/abs/2011.04672v2}{arXiv:2011.04672v2}.

\bibitem{Horodecki2005}
Micha{\l} Horodecki, Jonathan Oppenheim, and Andreas Winter.
\newblock ``Partial quantum information''.
\newblock \href{https://dx.doi.org/10.1038/nature03909}{Nature {\bf 436}, 673--676}~(2005).

\bibitem{Horodecki:2007:107}
Michal Horodecki, Jonathan Oppenheim, and Andreas Winter.
\newblock ``Quantum state merging and negative information''.
\newblock \href{https://dx.doi.org/10.1007/s00220-006-0118-x}{Communications in Mathematical Physics {\bf 269}, 107--136}~(2007).

\bibitem{BCCRR10}
Mario Berta, Matthias Christandl, Roger Colbeck, Joseph~M. Renes, and Renato Renner.
\newblock ``The uncertainty principle in the presence of quantum memory''.
\newblock \href{https://dx.doi.org/10.1038/nphys1734}{Nature Physics {\bf 6}, 659--662}~(2010).

\bibitem{Rio2011}
Lidia~del Rio, Johan Aberg, Renato Renner, Oscar Dahlsten, and Vlatko Vedral.
\newblock ``The thermodynamic meaning of negative entropy''.
\newblock \href{https://dx.doi.org/10.1038/nature10123}{Nature {\bf 474}, 61--63}~(2011).

\bibitem{Renner2005}
Renato Renner.
\newblock ``Security of quantum key distribution''.
\newblock PhD thesis.
\newblock ETH Zurich.
\newblock ~(2005).
\newblock  \href{http://arxiv.org/abs/quant-ph/0512258}{arXiv:quant-ph/0512258}.

\bibitem{KRS09}
Robert Koenig, Renato Renner, and Christian Schaffner.
\newblock ``The operational meaning of min- and max-entropy''.
\newblock \href{https://dx.doi.org/10.1109/TIT.2009.2025545}{IEEE Transactions on Information Theory {\bf 55}, 4337--4347}~(2009).

\bibitem{TCR09}
Marco Tomamichel, Roger Colbeck, and Renato Renner.
\newblock ``A fully quantum asymptotic equipartition property''.
\newblock \href{https://dx.doi.org/10.1109/TIT.2009.2032797}{IEEE Transactions on Information Theory {\bf 55}, 5840--5847}~(2009).

\bibitem{MDSFT13}
Martin {M\"uller}-Lennert, Fr\'ed\'eric Dupuis, Oleg Szehr, Serge Fehr, and Marco Tomamichel.
\newblock ``On quantum {R\'enyi} entropies: a new generalization and some properties''.
\newblock \href{https://dx.doi.org/10.1063/1.4838856}{Journal of Mathematical Physics {\bf 54}, 122203}~(2013).
\newblock  \href{http://arxiv.org/abs/1306.3142}{arXiv:1306.3142}.

\bibitem{TBH13}
Marco Tomamichel, Mario Berta, and Masahito Hayashi.
\newblock ``Relating different quantum generalizations of the conditional {R\'enyi} entropy''.
\newblock \href{https://dx.doi.org/10.1063/1.4892761}{Journal of Mathematical Physics {\bf 55}, 082206}~(2014).
\newblock  \href{http://arxiv.org/abs/1311.3887}{arXiv:1311.3887}.

\bibitem{HT16}
Masahito Hayashi and Marco Tomamichel.
\newblock ``Correlation detection and an operational interpretation of the {R}ényi mutual information''.
\newblock \href{https://dx.doi.org/10.1063/1.4964755}{Journal of Mathematical Physics {\bf 57}, 102201}~(2016).

\bibitem{LY1999}
Elliott~H. Lieb and Jakob Yngvason.
\newblock ``The physics and mathematics of the second law of thermodynamics''.
\newblock \href{https://dx.doi.org/https://doi.org/10.1016/S0370-1573(98)00082-9}{Physics Reports {\bf 310}, 1--96}~(1999).

\bibitem{LY1998}
Elliott~H. Lieb and Jakob Yngvason.
\newblock ``A guide to entropy and the second law of thermodynamics''.
\newblock \href{https://dx.doi.org/10.1007/978-3-662-10018-9_19}{Pages 353--363}.
\newblock Springer Berlin Heidelberg. Berlin, Heidelberg~(2004).
\newblock  \href{http://arxiv.org/abs/math-ph/9805005}{arXiv:math-ph/9805005}.

\bibitem{GT2021}
Gilad Gour and Marco Tomamichel.
\newblock ``Entropy and relative entropy from information-theoretic principles''.
\newblock \href{https://dx.doi.org/10.1109/TIT.2021.3078337}{IEEE Transactions on Information Theory {\bf 67}, 6313--6327}~(2021).

\bibitem{HH99}
Micha\l{} Horodecki and Pawe\l{} Horodecki.
\newblock ``Reduction criterion of separability and limits for a class of distillation protocols''.
\newblock \href{https://dx.doi.org/10.1103/PhysRevA.59.4206}{Physical Review A {\bf 59}, 4206--4216}~(1999).

\bibitem{Chitambar_2019}
Eric Chitambar and Gilad Gour.
\newblock ``Quantum resource theories''.
\newblock \href{https://dx.doi.org/10.1103/revmodphys.91.025001}{Reviews of Modern Physics {\bf 91}, 025001}~(2019).

\bibitem{HHH+2003}
Micha\l{} Horodecki, Karol Horodecki, Pawe\l{} Horodecki, Ryszard Horodecki, Jonathan Oppenheim, Aditi Sen(De), and Ujjwal Sen.
\newblock ``Local information as a resource in distributed quantum systems''.
\newblock \href{https://dx.doi.org/10.1103/PhysRevLett.90.100402}{Physical Review Letters {\bf 90}, 100402}~(2003).

\bibitem{HHO2003}
Micha\l{} Horodecki, Pawe\l{} Horodecki, and Jonathan Oppenheim.
\newblock ``Reversible transformations from pure to mixed states and the unique measure of information''.
\newblock \href{https://dx.doi.org/10.1103/PhysRevA.67.062104}{Physical Review A {\bf 67}, 062104}~(2003).

\bibitem{BGG22}
Sarah Brandsen, Isabelle~Jianing Geng, and Gilad Gour.
\newblock ``What is entropy? {A} perspective from games of chance''.
\newblock \href{https://dx.doi.org/10.1103/PhysRevE.105.024117}{Physical Review E {\bf 105}, 024117}~(2022).

\bibitem{U71}
Armin Uhlmann.
\newblock ``S\"atze \"uber dichtematrizen''.
\newblock Wissenschaftliche Zeitschrift der Karl-Marx-Universitat Leipzig. Mathematisch-naturwissenschaftliche Reihe {\bf 20}, 633--653~(1971).

\bibitem{U72}
Armin Uhlmann.
\newblock ``Endlich-dimensionale dichtematrizen {I}''.
\newblock Wissenschaftliche Zeitschrift der Karl-Marx-Universitat Leipzig. Mathematisch-naturwissenschaft-liche Reihe {\bf 21}, 421--452~(1972).

\bibitem{U73}
Armin Uhlmann.
\newblock ``Endlich-dimensionale dichtematrizen {II}''.
\newblock Wissenschaftliche Zeitschrift der Karl-Marx-Universitat Leipzig Mathematisch-naturwissenschaft-liche Reihe {\bf 22}, 139--177~(1973).

\bibitem{BHNOW15}
Fernando G. S.~L. Brand\~{a}o, Michal Horodecki, Nelly Ng, Jonathan Oppenheim, and Stephanie Wehner.
\newblock ``The second laws of quantum thermodynamics''.
\newblock \href{https://dx.doi.org/10.1073/pnas.1411728112}{Proceedings of the National Academy of Sciences {\bf 112}, 3275--3279}~(2015).

\bibitem{BGWG21}
Sarah Brandsen, Isabelle~J. Geng, Mark~M. Wilde, and Gilad Gour.
\newblock ``Quantum conditional entropy from information-theoretic principles''~(2021).
\newblock  \href{http://arxiv.org/abs/2110.15330}{arXiv:2110.15330}.

\bibitem{Vempati2022unitaloperations}
Mahathi Vempati, Saumya Shah, Nirman Ganguly, and Indranil Chakrabarty.
\newblock ``A-unital {O}perations and {Q}uantum {C}onditional {E}ntropy''.
\newblock \href{https://dx.doi.org/10.22331/q-2022-02-02-641}{{Quantum} {\bf 6}, 641}~(2022).

\bibitem{PR22}
Christopher Portmann and Renato Renner.
\newblock ``Security in quantum cryptography''.
\newblock \href{https://dx.doi.org/10.1103/RevModPhys.94.025008}{Reviews of Modern Physics {\bf 94}, 025008}~(2022).

\bibitem{Beckman_2001}
David Beckman, Daniel Gottesman, Michael~A. Nielsen, and John Preskill.
\newblock ``Causal and localizable quantum operations''.
\newblock \href{https://dx.doi.org/10.1103/physreva.64.052309}{Physical Review A {\bf 64}, 052309}~(2001).

\bibitem{ESW02}
T.~Eggeling, D.~Schlingemann, and Reinhard~F. Werner.
\newblock ``Semicausal operations are semilocalizable''.
\newblock \href{https://dx.doi.org/10.1209/epl/i2002-00579-4}{Europhysics Letters {\bf 57}, 782--788}~(2002).

\bibitem{Piani_2006}
Marco Piani, Michal Horodecki, Pawel Horodecki, and Ryszard Horodecki.
\newblock ``Properties of quantum nonsignaling boxes''.
\newblock \href{https://dx.doi.org/10.1103/physreva.74.012305}{Physical Review A {\bf 74}, 012305}~(2006).

\bibitem{GGHKLV18}
Gilad Gour, A.~Grudka, M.~Horodecki, W.~Klobus, J.~Lodyga, and V.~Narasimhachar.
\newblock ``The conditional uncertainty principle''.
\newblock \href{https://dx.doi.org/10.1103/PhysRevA.97.042130}{Physical Review A {\bf 97}, 042130}~(2018).

\bibitem{W89}
Reinhard~F. Werner.
\newblock ``Quantum states with {Einstein-Podolsky-Rosen} correlations admitting a hidden-variable model''.
\newblock \href{https://dx.doi.org/10.1103/PhysRevA.40.4277}{Physical Review A {\bf 40}, 4277--4281}~(1989).

\bibitem{GT20}
Gilad Gour and Marco Tomamichel.
\newblock ``Optimal extensions of resource measures and their applications''.
\newblock \href{https://dx.doi.org/10.1103/physreva.102.062401}{Physical Review A {\bf 102}, 062401}~(2020).

\bibitem{U62}
Hisaharu Umegaki.
\newblock ``Conditional expectations in an operator algebra {IV} (entropy and information)''.
\newblock \href{https://dx.doi.org/10.2996/kmj/1138844604}{Kodai Mathematical Seminar Reports {\bf 14}, 59--85}~(1962).

\bibitem{Belavkin1982}
V.~P. Belavkin and P.~Staszewski.
\newblock ``C*-algebraic generalization of relative entropy and entropy''.
\newblock Annales de l'I.H.P. Physique th\'eorique {\bf 37}, 51--58~(1982).
\newblock  url:~\url{http://www.numdam.org/item/AIHPA_1982__37_1_51_0/}.

\bibitem{P85}
D\'enes Petz.
\newblock ``{Quasi-entropies for States of a von Neumann Algebra}''.
\newblock \href{https://dx.doi.org/10.2977/prims/1195178929}{Publications of the Research Institute for Mathematical Sciences {\bf 21}, 787--800}~(1985).

\bibitem{P86}
D\'enes Petz.
\newblock ``Quasi-entropies for finite quantum systems''.
\newblock \href{https://dx.doi.org/10.1016/0034-4877(86)90067-4}{Reports in Mathematical Physics {\bf 23}, 57--65}~(1986).

\bibitem{WWY14}
Mark~M. Wilde, Andreas Winter, and Dong Yang.
\newblock ``Strong converse for the classical capacity of entanglement-breaking and {Hadamard} channels via a sandwiched {R\'enyi} relative entropy''.
\newblock \href{https://dx.doi.org/10.1007/s00220-014-2122-x}{Communications in Mathematical Physics {\bf 331}, 593--622}~(2014).

\bibitem{AD15}
Koenraad M.~R. Audenaert and Nilanjana Datta.
\newblock ``$\alpha$-$z$-{R\'enyi} relative entropies''.
\newblock \href{https://dx.doi.org/10.1063/1.4906367}{Journal of Mathematical Physics {\bf 56}, 022202}~(2015).
\newblock  \href{http://arxiv.org/abs/1310.7178}{arXiv:1310.7178}.

\bibitem{M13}
Keiji Matsumoto.
\newblock ``A new quantum version of $f$-divergence''~(2013).
\newblock  \href{http://arxiv.org/abs/1311.4722}{arXiv:1311.4722}.

\bibitem{Matsumoto2018}
Keiji Matsumoto.
\newblock ``A new quantum version of f-divergence''.
\newblock In Masanao Ozawa, Jeremy Butterfield, Hans Halvorson, Mikl{\'o}s R{\'e}dei, Yuichiro Kitajima, and Francesco Buscemi, editors, Reality and Measurement in Algebraic Quantum Theory.
\newblock \href{https://dx.doi.org/10.1007/978-981-13-2487-1_10}{Volume 261, pages 229--273}.
\newblock Singapore~(2018). Springer Singapore.

\bibitem{SW13}
Naresh Sharma and Naqueeb~Ahmad Warsi.
\newblock ``Fundamental bound on the reliability of quantum information transmission''.
\newblock \href{https://dx.doi.org/10.1103/PhysRevLett.110.080501}{Physical Review Letters {\bf 110}, 080501}~(2013).

\bibitem{Bei13}
Salman Beigi.
\newblock ``Sandwiched {R\'enyi} divergence satisfies data processing inequality''.
\newblock \href{https://dx.doi.org/10.1063/1.4838855}{Journal of Mathematical Physics {\bf 54}, 122202}~(2013).
\newblock  \href{http://arxiv.org/abs/1306.5920}{arXiv:1306.5920}.

\bibitem{NISO}
NISO.
\newblock ``Credit – contributor roles taxonomy''.
\newblock \url{https://credit.niso.org/}, Accessed 2024-09-04.

\bibitem{W18thesis}
Xin Wang.
\newblock ``Semidefinite optimization for quantum information''.
\newblock PhD thesis.
\newblock University of Technology Sydney.
\newblock Centre for Quantum Software and Information, Faculty of Engineering and Information Technology~(2018).
\newblock  url:~\url{http://hdl.handle.net/10453/127996}.

\bibitem{Gour5}
Gilad {Gour}.
\newblock ``Comparison of quantum channels by superchannels''.
\newblock \href{https://dx.doi.org/10.1109/TIT.2019.2907989}{IEEE Transactions on Information Theory {\bf 65}, 5880--5904}~(2019).

\bibitem{SDGWW21}
Robert Salzmann, Nilanjana Datta, Gilad Gour, Xin Wang, and Mark~M. Wilde.
\newblock ``Symmetric distinguishability as a quantum resource''.
\newblock \href{https://dx.doi.org/10.1088/1367-2630/ac14aa}{New Journal of Physics {\bf 23}, 083016}~(2021).

\bibitem{KR11}
Robert K\"onig and Renato Renner.
\newblock ``Sampling of min-entropy relative to quantum knowledge''.
\newblock \href{https://dx.doi.org/10.1109/TIT.2011.2146730}{IEEE Transactions on Information Theory {\bf 57}, 4760--4787}~(2011).

\bibitem{Gour2023}
Gilad Gour.
\newblock ``Resources of the quantum world ({V}olume~1)''~(2024).
\newblock  \href{http://arxiv.org/abs/2402.05474v1}{arXiv:2402.05474v1}.

\bibitem{D09}
Nilanjana Datta.
\newblock ``Min- and max-relative entropies and a new entanglement monotone''.
\newblock \href{https://dx.doi.org/10.1109/TIT.2009.2018325}{IEEE Transactions on Information Theory {\bf 55}, 2816--2826}~(2009).

\end{thebibliography}
\bibliographystyle{quantum}


\appendix

\section{Proofs of properties of entropy}

\label{app:entropy-props}

Let us employ the postulates in~\eqref{eq:entropy-def-1}--\eqref{eq:entropy-def-3} to derive some basic properties of every entropy. The proofs given here are similar to those given for the classical case in~\cite[Lemma~3]{GT2021}. First, let us recall an often overlooked fact, that the number 1 is a trivial quantum state of dimension one. Indeed, the underlying Hilbert space is of dimension one, so that the only possible density ``operator'' for this space is the number 1. We now prove that the entropy of this state is equal to zero, as a direct consequence of the additivity postulate. Consider that, for a state $\rho$,
\begin{equation}
\H(\rho) = \H(\rho \otimes 1) = \H(\rho) + \H(1)  
 \ \Rightarrow \   \H(1)  = 0. \label{eq:entr-1-is-0}
\end{equation}
The first equality follows as a trivial assertion, because $\rho = \rho \otimes 1$. The next equality is a consequence of the additivity postulate. The conclusion that $\H(1) = 0$ then follows by canceling $\H(\rho)$ on both sides of the equality.

It follows that the entropy of every pure state $|\psi\rangle\!\langle\psi|$ is equal to zero, as a consequence of~\eqref{eq:entr-1-is-0}. Indeed, the state vector $|\psi\rangle$ is an isometry from the trivial Hilbert space of dimension one to the Hilbert space that contains~$|\psi\rangle$. As such, it follows that $1 \succ |\psi\rangle\!\langle\psi|$ and $1 \prec |\psi\rangle\!\langle\psi|$ according to the definitions in~\eqref{eq:final-maj-def}--\eqref{eq:iso-conditions-majorization} and~\eqref{eq:entropy-def-1}--\eqref{eq:entropy-def-3}. So we conclude that
\begin{equation}
\H(|\psi\rangle\!\langle\psi|)=0 \label{eq:pure-state-zero-ent}
\end{equation}
for every pure state $|\psi\rangle\!\langle\psi|$.

Next we conclude that the entropy of every state is non-negative. This follows because there exists a unital channel taking a pure state $|\psi\rangle\!\langle\psi|$ to an arbitrary state~$\rho$, so that $|\psi\rangle\!\langle\psi| \succ \rho$. Indeed, let $\rho = \sum_x p_x |\psi_x\rangle\!\langle \psi_x|$ be a pure state decomposition of $\rho$. Let $U_x$ be a unitary that transforms $|\psi\rangle$ to $|\psi_x\rangle$. Thus, $(\cdot) \to \sum_x p_x U_x (\cdot) U_x^\dag$ is a unital channel that transforms $|\psi\rangle\!\langle\psi|$ to $\rho$. By applying~\eqref{eq:entropy-def-2} and~\eqref{eq:pure-state-zero-ent}, we conclude that
\begin{equation}
\H(\rho) \geq 0
\label{eq:all-states-non-neg}
\end{equation}
for every state $\rho$.

For a fixed dimension $d$, the uniform state $\u^{(d)} = I / d$ is the state of maximal entropy, which follows from~\eqref{eq:entropy-def-2} and because every $d$-dimensional state $\rho$ can be taken to $\u^{(d)}$ by means of the unital channel $(\cdot) \to \operatorname{Tr}[(\cdot)] \u^{(d)}$. Thus,
\begin{equation}
\H(\u^{(d)}) \geq \H(\rho)
\label{eq:uniform-largest-ent}
\end{equation}
for every state $\rho$ of dimension $d$.

Let us prove that $\H(\u^{(2)}) > 0$. By the inequality in~\eqref{eq:all-states-non-neg} and the assumption that entropy is not equal to the zero function, there exists a state $\rho$ of some dimension~$d$ such that $\H(\rho) > 0$. By applying~\eqref{eq:uniform-largest-ent}, we conclude that $\H(\u^{(d)}) \geq \H(\rho)$. By taking $n$ sufficiently large (i.e., such that $2^n > d$), we conclude that $\u^{(d)} \succ (\u^{(2)})^{\otimes n}$, which implies that $\H((\u^{(2)})^{\otimes n}) > \H(\u^{(d)})$. Finally, by applying the additivity postulate, we conclude that $ n \H(\u^{(2)}) = \H((\u^{(2)})^{\otimes n})$. Combining all inequalities, we conclude that $\H(\u^{(2)}) > 0$.

After establishing the normalization convention that $\H(\u^{(2)})=1$, the proof that $\H(\u^{(d)}) = \log_2 d$ is the same as that given in the proof of~\cite[Lemma~3]{GT2021}.

\section{Choi matrix formulations}

In this appendix, we provide Choi matrix formulations of conditionally unital channels and conditionally mixing channels. These formulations are useful conceptually and also for numerical optimization if one wishes to use semi-definite programming to optimize over these sets of channels. See~\cite{Watrous2018,W18thesis,KW20book} for examples and discussions of semi-definite programming in quantum information.

Before giving these formulations, let us first recall the definition of the Choi matrix. The Choi matrix $J^{\mathcal{M}}_{CD}$ of a linear map $\mathcal{M}_{C \rightarrow D}$ is defined by its action on the maximally entangled operator, as follows:
\begin{align}
    J^{\mathcal{M}}_{CD} \coloneqq \sum_{i, j =1}^{|C|} |i\rangle\!\langle j| \otimes \mathcal{M}(|i\rangle\!\langle j| ).
\end{align}
In general, the matrix $J^{\mathcal{M}}_{CD}$ is the Choi matrix of a trace-preserving map  if and only if
\begin{equation}
\operatorname{Tr}_{D}[J^{\mathcal{M}}_{CD}] = I_C,
    \label{eq:Choi-TP}
\end{equation}
and it is the Choi matrix of a completely positive map if and only if
\begin{equation}
    J^{\mathcal{M}}_{CD} \geq 0.
        \label{eq:Choi-CP}
\end{equation}
Thus, $J^{\mathcal{M}}_{CD}$ is the Choi matrix of a quantum channel if and only if~\eqref{eq:Choi-TP} and~\eqref{eq:Choi-CP} both hold.

\subsection{Choi matrix formulation of conditional unitality}

\label{app:Choi-cond-unital}

Here we prove that a conditionally unital channel has an equivalent formulation in terms of its Choi matrix. We then see that a channel being conditionally unital represents a semi-definite constraint on its Choi matrix, which is convenient for optimizing over such channels using semi-definite programming.

\begin{lemma}
\label{lem:cond-unital-choi}
A channel
$\mathcal{N}_{AB \rightarrow \tilde{A}B'}$ is conditionally unital if and only if its Choi matrix $J^{\mathcal{N}}_{AB\tilde{A}B'}$ satisfies 
\begin{equation}
    J^{\mathcal{N}}_{B \tilde{A} B'} = J^{\mathcal{N}}_{BB'} \otimes \mathbf{u}_{\tilde{A}},
    \label{eq:choi-matrix-cond-unital}
\end{equation}
where
\begin{align}
J^{\mathcal{N}}_{B \tilde{A} B'} & \coloneqq \operatorname{Tr}_{A}[J^{\mathcal{N}}_{AB\tilde{A}B'}],\\
J^{\mathcal{N}}_{BB'} & \coloneqq  \operatorname{Tr}_{A\tilde{A}}[J^{\mathcal{N}}_{AB\tilde{A}B'}],
\end{align}
are marginals of $J^{\mathcal{N}}_{AB\tilde{A}B'}$ and we employ this kind of shorthand in what follows.
\end{lemma}

\begin{proof}
We begin by proving that the channel $\mathcal{N}_{AB \rightarrow \tilde{A}B'}$ is conditionally unital if its Choi matrix has the form in~\eqref{eq:choi-matrix-cond-unital}. To this end, let $\rho_B\in \mathfrak{D}(B)$ and  consider that
\begin{align}
    & \mathcal{N}\left(\mathbf{u} _{A} \otimes \rho_{B}\right) \notag \\
    &= \operatorname{Tr}_{AB}\!\left[ J^{\mathcal{N}}_ {AB\tilde{A}B'} \left( \mathbf{u}_{A} \otimes \left(\rho_{B}\right)^{T} \otimes I_{\tilde{A}B'} \right)\right] \\
    &= \frac{1}{|A|} \operatorname{Tr}_{B}\!\left[ J^{\mathcal{N}}_{B\tilde{A}B'} \left( \left(\rho_{B}\right)^{T} \otimes I_{\tilde{A}B'} \right)\right] \\
    & = \frac{1}{|A|} \operatorname{Tr}_{B}\!\left[\mathbf{u}_{\tilde{A}} \otimes  J^{\mathcal{N}}_{BB'} \left( \left(\rho_{B}\right)^{T} \otimes I_{\tilde{A}B'} \right)\right] \\
    & = \mathbf{u}_{\tilde{A}} \otimes \frac{1}{|A|} \operatorname{Tr}_{B}\!\left[  J^{\mathcal{N}}_{BB'} \left( \left(\rho_{B}\right)^{T} \otimes I_{B'} \right)\right] .
\end{align}
The first equality is a well known identity in quantum information, indicating how the action of a channel on an input state can be understood in terms of a propagation formula using its Choi matrix (see Eq.~(2.66) of~\cite{Watrous2018}). 
The third equality above follows from applying~\eqref{eq:choi-matrix-cond-unital}, which holds by assumption.
It then follows that $\frac{1}{|A|} \operatorname{Tr}_{B}\!\left[  J^{\mathcal{N}}_{BB'} \left( \left(\rho_{B}\right)^{T} \otimes I_{B'} \right)\right]$ is a state on the output system $B'$ (i.e., the choice for the state $\sigma_{B'}$ in~\eqref{eq:cond-unital-def}) because $\frac{1}{|A|}   J^{\mathcal{N}}_{BB'}$ is the Choi operator of the channel
\begin{equation}
\tau_B \to (\operatorname{Tr}_{\tilde{A}} \circ \mathcal{N}_{AB\to\tilde{A}B'})(\mathbf{u}_{A} \otimes \tau_B ).
\end{equation}

We now prove that the Choi matrix of $\mathcal{N}_{AB\rightarrow\tilde
{A}B^{\prime}}$ has the form in~\eqref{eq:choi-matrix-cond-unital} if
$\mathcal{N}_{AB\rightarrow\tilde{A}B^{\prime}}$ is conditionally unital.
Recall from~\eqref{eq:cond-unital-def} that the defining property of a conditionally unital channel is that
for every state $\rho_{B}$, there exists a state~$\sigma_{B^{\prime}}$ such
that%
\begin{equation}\label{eq14}
\mathcal{N}_{AB\to\tA B'}(\mathbf{u}_{A}\otimes\rho_{B})=\mathbf{u}_{\tilde{A}}\otimes
\sigma_{B^{\prime}}.
\end{equation}
By taking the trace over $\tA$ on both sides of the equation above, we find that
\begin{align}
\sigma_{B^{\prime}}&=\mathcal{N}_{AB\to B'}(\mathbf{u}_{A}\otimes\rho_{B})\\
&=\tr_{AB}\!\left[  J_{ABB^{\prime}}^{\mathcal{N}%
}\left(  \mathbf{u}_{A}\otimes\left(  \rho_{B}\right)  ^{T}\otimes
I_{B^{\prime}}\right)  \right]\\
&=\frac1{|A|}\tr_{B}\!\left[  J_{BB^{\prime}}^{\mathcal{N}%
}\left( \left(  \rho_{B}\right)  ^{T}\otimes
I_{B^{\prime}}\right)  \right]\;.
\end{align}
On the other hand, observe that
\begin{align}
&  \mathcal{N}_{AB\to\tA B'}(\mathbf{u}_{A}\otimes\rho_{B})\nonumber\\
&  =\operatorname{Tr}_{AB}\!\left[  J_{AB\tilde{A}B^{\prime}}^{\mathcal{N}
}\left(  \mathbf{u}_{A}\otimes\left(  \rho_{B}\right)  ^{T}\otimes
I_{\tilde{A}B^{\prime}}\right)  \right]  \\
&  =\frac{1}{|A|}\operatorname{Tr}_{B}\!\left[  J_{B\tilde{A}B^{\prime}
}^{\mathcal{N}}\left(  \left(  \rho_{B}\right)  ^{T}\otimes I_{\tilde
{A}B^{\prime}}\right)  \right]  \;.
\end{align}
Therefore, from the two expressions above for $\sigma_B$ and $\mathcal{N}(\mathbf{u}_{A}\otimes\rho_{B})$, we conclude that~\eqref{eq14} can be expressed as
\begin{multline}\label{eq20}
\operatorname{Tr}_{B}\!\left[  J_{B\tilde{A}B^{\prime}
}^{\mathcal{N}}\left(  \left(  \rho_{B}\right)  ^{T}\otimes I_{\tilde
{A}B^{\prime}}\right)  \right]
\\
= \tr_{B}\!\left[ (\u_{\tA}\otimes J_{BB^{\prime}}^{\mathcal{N}%
})\left( \left(  \rho_{B}\right)  ^{T}\otimes
I_{\tA B^{\prime}}\right)  \right]\;,
\end{multline}
holding for all $\rho_B \in \mathfrak{D}(B)$.
Let
\begin{equation}
\eta_{B\tA B'}\eqdef J_{B\tilde{A}B^{\prime}
}^{\mathcal{N}}-\u_{\tA}\otimes J_{BB^{\prime}}^{\mathcal{N}},
\end{equation}
and observe that~\eqref{eq20} can be written as
\begin{equation}
\label{eq21}
\tr_B\!\left[\eta_{B\tA B'}\left(\left(  \rho_{B}\right)  ^{T}\otimes
I_{\tA B^{\prime}}\right)\right]=0\;\;\quad\forall\;\rho_B\in\mathfrak{D}(B).
\end{equation}
Due to the existence of bases of density operators that span the space of
linear operators acting on ${B}$ (see Example~2.7 of~\cite{Watrous2018}),
we conclude from~\eqref{eq21} that for every Hermitian
operator $X_{B}\in\mathfrak{L}(B)$, we have
\begin{equation}
\operatorname{Tr}_{B}\!\left[  \eta_{B\tilde{A}B^{\prime}%
}\left(  X_{B}\otimes I_{\tilde{A}B^{\prime}}\right)  \right]
=0\;.
\end{equation}
Note that by multiplying both sides of the equation above by an arbitrary element $Y_{\tA B'}\in\mathfrak{L}(\tA B')$ and taking the trace, we get that
\begin{equation}
\tr\big[  \eta_{B\tilde{A}B^{\prime}%
}\left(  X_{B}\otimes Y_{\tilde{A}B^{\prime}}\right)  \big]
=0\;.
\end{equation}
Since the equation above holds for all $X_B\in\mathfrak{L}(B)$ and all $Y_{\tA B'}\in\mathfrak{L}(\tA B')$ it also holds for all linear combinations of matrices of the form $X_{B}\otimes Y_{\tilde{A}B^{\prime}}$. Since matrices of the form $X_{B}\otimes Y_{\tilde{A}B^{\prime}}$ span the whole space $\mathfrak{L}(B\tA B')$ we conclude that
\begin{equation}
\tr\big[  \eta_{B\tilde{A}B^{\prime}%
}W_{B\tilde{A}B^{\prime}} \big]
=0\quad\forall\;W_{B\tilde{A}B^{\prime}}\in\mathfrak{L}(B\tA B')\;.
\end{equation}
Therefore, we conclude that $\eta_{B\tilde{A}B'}=0$, which is equivalent to~\eqref{eq:choi-matrix-cond-unital}. 
\end{proof}

\subsection{Choi matrix formulation of conditionally mixing channels}

\label{app:locally-balanced-choi}

Here we provide a Choi matrix formulation for conditionally mixing channels. Before doing so, let us first recall~\cite[Eq.~(40)]{Piani_2006}, which indicates how semi-causality can be expressed in terms of the Choi matrix.

\begin{lemma}[\cite{Piani_2006,Gour5}]
\label{lem:semi-causal}
The channel $\mathcal{N}_{AB \rightarrow AB'}$ is $A\not\to B'$ semi-causal if and only if the marginals of its Choi matrix~$J^{\mathcal{N}}_{AB\tA B'}$ satisfy
\begin{align}
    J^{\mathcal{N}}_{ABB'} = \mathbf{u}_{A} \otimes J^{\mathcal{N}}_{BB'}.
\end{align}
\end{lemma}

As a direct corollary of~\eqref{eq:Choi-TP}--\eqref{eq:Choi-CP} and Lemmas~\ref{lem:cond-unital-choi} and~\ref{lem:semi-causal}, we have the following simple semi-definite characterization of a conditionally mixing channel in terms of its Choi matrix.

\begin{corollary}
\label{cor:CUSC-Choi-conds}
A matrix $J_{AB\tilde{A}B^{\prime}}^{\mathcal{N}}$ is the Choi matrix of a conditionally mixing channel $\mathcal{N}_{AB \rightarrow \tilde{A}B'}$ if and only if
\begin{align}
J_{AB\tilde{A}B^{\prime}}^{\mathcal{N}} & \geq 0, \\
\operatorname{Tr}_{\tilde{A}B^{\prime}}[J_{AB\tilde{A}B^{\prime}}^{\mathcal{N}}] & = I_{AB}, \\ 
J^{\mathcal{N}}_{ABB'} & = \mathbf{u}_{A} \otimes J^{\mathcal{N}}_{BB'},
\label{eq:LB-SDP-no-sig}\\
J^{\mathcal{N}}_{B \tilde{A} B'} & = J^{\mathcal{N}}_{BB'} \otimes \mathbf{u}_{\tilde{A}}.
\label{eq:LB-SDP-cond-unit}
\end{align}
\end{corollary}

Corollary~\ref{cor:CUSC-Choi-conds} simplifies the task of optimizing an objective function over the set of conditionally mixing channels by making use of semi-definite programming.
Additionally, by inspecting~\eqref{eq:LB-SDP-no-sig}--\eqref{eq:LB-SDP-cond-unit}, we see that the properties of these channels are invariant under the exchange of the input system $A$ and the output system $\tilde{A}$.

\section{Quantum conditional majorization reduces to classical conditional majorization}

\label{app:q-cond-maj-reduces-to-classical-cond-maj}

In~\cite{GGHKLV18}, a pre-order $\succ_{C}$ between two joint probability distributions (i.e., classical states) is called \emph{classical conditional majorization}. In this appendix, we prove that the approach to conditional majorization defined in Section~\ref{sec:q-cond-maj} is a generalization of classical conditional majorization. To this end, we demonstrate that quantum conditional majorization is equivalent to the classical conditional majorization when we restrict the underlying bipartite quantum states to classical states. 

Suppose the system $A\eqdef X$ is classical. We call a channel $\mathcal{M}_{XB\to XB'}$ \emph{conditionally doubly stochastic} (CDS for short) if it has the following form: 
\begin{equation}\label{c1}
    \mathcal{M}_{XB\to XB'}=\sum_{j}\mathcal{D}_{X\to X}^{(j)}\otimes\mathcal{F}_{B\to B'}^{(j)},
\end{equation}
where each $\mathcal{D}^{(j)}$ is a classical doubly stochastic channel and each $\mathcal{F}_{B\to B'}^{(j)}$ is a completely positive (CP) map such that $\sum_j \mathcal{F}_{B\to B'}^{(j)}$ is a quantum channel. CDS channels with classical $B$ and $B'$ systems were introduced and studied in~\cite{GGHKLV18}, and CDS channels with quantum $B$ and $B'$ systems were introduced and studied in~\cite{SDGWW21}.

\begin{definition}[\cite{GGHKLV18}]
Let $\rho_{XY}$ and $\sigma_{XY'}$ be classical bipartite states such that $A=X$, $B=Y$, and $B'=Y'$ are all classical systems. We say that $\rho_{XY}$ conditionally majorizes $\sigma_{XY'}$ and write $\rho_{XY} \succ_{C} \sigma_{XY'}$ if there exists a CDS channel $\mathcal{N}_{XY\to XY'}$ such that $\sigma_{XY'}=\mathcal{N}_{XY\to XY'}(\rho_{XY}) $.
\end{definition}

\begin{theorem}
\label{lem_red}
Suppose all the systems involved are classical with $A\eqdef X$, $B\eqdef Y$, and $B'\eqdef Y'$. Then a quantum channel $\mathcal{N}_{XY \rightarrow XY'}$ is a CDS channel if and only if it is conditionally mixing.
\end{theorem}

\begin{proof}
Suppose first that $\mathcal{N}_{XY \rightarrow XY'}$ is a CDS channel, having the form of~\eqref{c1}.
Then for every input of the form $\mathbf{u}_{A} \otimes \sigma_{B}$, it follows that 
\begin{align}
    & \mathcal{N}_{AB\to AB'}(\mathbf{u}_{A} \otimes \sigma_{B}) \notag \\
    &=\sum_{j}\mathcal{D}_{A}^{(j)}(\mathbf{u}_{A}) \otimes\mathcal{F}_{B\to B'}^{(j)}(\sigma_{B})  \\
    &= \mathbf{u}_{A} \otimes \sum_{j}\mathcal{F}_{B\to B'}^{(j)}(\sigma_{B})
\end{align}
and therefore $\mathcal{N}$ is conditionally unital. (Note that the third line of the above follows from the fact that $\mathcal{D}^{(j)}$ is doubly stochastic, so that $\mathcal{D}^{(j)}(\mathbf{u}_{A}) = \mathbf{u}_{A}$ for all $j$.) Next, we show that $\mathcal{N}$ is semi-causal. Let $\mathcal{T}_{A}$ be an arbitrary classical channel. Then consider that
\begin{align}
& \operatorname{Tr}_{A}\circ  \mathcal{N} \circ \mathcal{T}_{A}  \notag \\
   & = \operatorname{Tr}_{A}\circ \!\left[ \sum_{j}(\mathcal{D}_{A}^{(j)} \circ \mathcal{T}_{A}) \otimes\mathcal{F}_{B\to B'}^{(j)} \right] \\
    &=   \sum_{j} (\operatorname{Tr}_{A}\circ \mathcal{D}_{A}^{(j)} \circ \mathcal{T}_{A})  \otimes\mathcal{F}_{B\to B'}^{(j)}  \\
     &=   \sum_{j} (\operatorname{Tr}_{A}\circ \mathcal{D}_{A}^{(j)} ) \otimes\mathcal{F}_{B\to B'}^{(j)}  \\
     &= \operatorname{Tr}_{A}\circ  \mathcal{N}_{AB \rightarrow AB'} .
\end{align}
Given that $\mathcal{N}$ is conditionally unital and semi-causal, it is by definition conditionally mixing.

Conversely, suppose $\mN_{XY\to XY'}$ is conditionally mixing. Since $\mN$ is semi-causal, it has the form given in~\eqref{eq:semi-causal-breakdown} (see also Figure~\ref{fig_semi-casual_channel}). Moreover, since systems $A$, $B$, and $B'$ are all classical, we can assume without loss of generality that system $R$ in~\eqref{eq:semi-causal-breakdown} is also classical. Hence, $\mN$ has the form
\be\label{formns}
\mN_{XY\to XY'}=\sum_{j=1}^k\mE^{(j)}_{X\to X}\otimes\mR^{(j)}_{Y\to Y'}\;,
\ee 
where $k\in\mathbb{N}$, and for each $j\in[k]$, $\mE^{(j)}\in\cptp(X\to X)$ and $\mR^{(j)}\in\cp(Y\to Y')$ with $\mR\eqdef\sum_j\mR^{(j)}\in\cptp(Y\to Y')$. The key idea of the proof is to decompose each $\mR^{(j)}$ into ``extreme" maps. We do that in the following way. We first denote by $\{r_{jy'|y}\}$ the components of the transition matrix of $\mR^{(j)}$. That is, for all $j\in[k]$ and $y\in[n]$
\be
\mR^{(j)}_{Y\to Y'}\left(|y\lr y|_Y\right)=\sum_{y'=1}^{n'}r_{jy'|y}|y'\lr y'|_{Y'}\;.
\ee
We also denote by $r_{y'|y}\eqdef\sum_{j\in[k]}r_{jy'|y}$ the components of the transition matrix of $\mR$, and by $q_{j|yy'}\eqdef r_{jy'|y}/r_{y'|y}$. Finally, for any $y\in[n]$ and $y'\in[n']$ we define the CP map $\mF^{(yy')}$ as
\be\label{5p215}
\mF^{(yy')}_{Y\to Y'}\left(|w\lr w|_{Y}\right)=\delta_{yw}r_{y'|y}|y'\lr y'|_{Y'}\;.
\ee
By definition, $\sum_{y\in[n]}\sum_{y'\in[n']}\mF^{(yy')}$ is a (CPTP) classical channel, and
\be
\mR^{(j)}_{Y\to Y'}=\sum_{y\in[n]}\sum_{y'\in[n']}q_{j|yy'}\mF^{(yy')}_{Y\to Y'}\;.
\ee
With these definitions we get that
\begin{align}\label{5p217}
\mN_{XY\to XY'}&=\sum_{y\in[n]}\sum_{y'\in[n']}\sum_{j\in[k]}q_{j|yy'}\mE^{(j)}_{X\to X}\otimes\mF^{(yy')}_{Y\to Y'}\nonumber\\
&=\sum_{y\in[n]}\sum_{y'\in[n']}\mD^{(yy')}_{X\to X}\otimes\mF^{(yy')}_{Y\to Y'}\;,
\end{align}
where 
\be
\mD_{X\to X'}^{(yy')}\eqdef\sum_{j\in[k]}q_{j|yy'}\mE^{(j)}_{X\to X}\;.
\ee
Observe that each $\mD_{X\to X'}^{(yy')}$ is a classical channel since a convex combination of classical channels is a classical channel. It is therefore left to show that each $\mD_{X\to X'}^{(yy')}$ is unital. To show that, we use the fact that $\mN$ is a conditionally unital channel.

Since $\mN$ is a conditionally unital channel,  for all $w\in[n]$ we have that
\be
\u_X\otimes\sigma_{Y'}^{(w)}=\mN_{XY\to XY'}\left(\u_X\otimes |w\lr w|_{Y}\right),
\ee
where $\sigma^{(w)}\in\md(Y')$ is some state (in fact, $\sigma_{Y'}^{(w)}=\mN_{XY\to Y'}\left(\u_X\otimes |w\lr w|_{Y}\right)$). Therefore, using the form of $\mN$ in~\eqref{5p217} we get
\begin{align}
&\u_X\otimes\sigma_{Y'}^{(w)}\nonumber\\
&=\sum_{y\in[n]}\sum_{y'\in[n']}\mD_{X\to X'}^{(yy')}\left(\u_X\right)\otimes\mF^{(yy')}_{Y\to Y'}\left(|w\lr w|_{Y}\right)\\
&=\sum_{y'\in[n']}r_{y'|w}\mD^{(wy')}_{X\to X}\left(\u_X\right)\otimes|y'\lr y'|_{Y'}\;,
\end{align}
where the last equality follows from~\eqref{5p215}.
By taking the expectation value $\la w'|(\cdot)|w'\ra$ of system $Y'$ on both sides of the equation above, we get that for all $w'\in[n']$
\be
\u^X\la w'|\sigma_{Y'}^{(w)}|w'\ra=r_{w|w'}\mD^{(ww')}_{X\to X}\left(\u_X\right)\;.
\ee
By taking the trace on both sides of the equation above we get that  $\la w'|\sigma_{Y'}^{(w)}|w'\ra=r_{w|w'}$ (recall that $\mD^{(ww')}_{X\to X}$ is a channel and in particular trace preserving). Therefore, the equation above implies that 
\be
\u_X=\mD^{(ww')}_{X\to X}\left(\u_X\right)\;;
\ee
that is, each $\mD^{(ww')}_{X\to X}$ is a doubly stochastic matrix.
\end{proof}

\begin{remark}
Observe that the theorem above implies in particular that
$\sigma_{XY'} \succ_X  \rho_{XY}$ if and only if $\sigma_{XY'} \succ_{C} \rho_{XY}$, where $\succ_{C}$ is the classical conditional majorization defined in~\cite{GGHKLV18}. In other words, the pre-orders $\succ_X$ and $\succ_C$ coincide in the classical domain.
\end{remark}

\section{Proofs of properties of quantum conditional entropy}

\label{app:cond-ent-proofs}

Suppose that $\rho_{AB}=\omega_A\otimes\tau_B$ is a product state. Let $\mathcal{E}_{B}$ be the completely randomizing channel acting on system~$B$, which traces out system $B$ and replaces with the uniform state $\u_B$. Then it follows that $\operatorname{id}_{A}\otimes\mathcal{E}_{B}$ is a conditionally unital channel and that
\begin{equation}
    (\operatorname{id}_{A}\otimes\mathcal{E}_{B})\left(\omega_A\otimes\tau_B\right)=\omega_A\otimes\u_{B} .
\end{equation}
Moreover, let $\mathcal{R}_{B}$ be the replacement channel on system~$B$, which takes every quantum state to a fixed state~$\tau_B$. Then, $\operatorname{id}_{A}\otimes\mathcal{R}_{B}$ is conditionally unital and
\begin{equation}
     (\operatorname{id}_{A}\otimes\mathcal{R}_{B})\left(\omega_A\otimes\u_B\right)=\omega_A\otimes\tau_{B}\ .
\end{equation}
Thus, 
\begin{equation}
    \omega_A\otimes\tau_B\succ_{A}\omega_A\otimes\u_B \text{ and } \omega_A\otimes\u_B\succ_{A}\omega_A\otimes\tau_B,
\end{equation}
which implies, by the monotonicity property of the conditional entropy $\H$, that
\begin{equation}
    \H(A|B)_{ \omega_A\otimes\tau_B}= \H(A|B)_{ \omega_A\otimes\u_B}.
\end{equation}
Thus, $\H(A|B)_{ \omega_A\otimes\tau_B}$ depends on $\omega_A$ only. Moreover, the function $\omega_A\mapsto \H(A|B)_{ \omega_A\otimes\tau_B}$ satisfies the two postulates of entropy in~\eqref{eq:entropy-def-2}--\eqref{eq:entropy-def-3} and therefore is an entropy of $\omega_A$.
Note that the same statement also holds when system $|B|$ is trivial, i.e., when $|B|=1$.

We now prove that $\H(A|B)$ is invariant under local isometric channels. It follows immediately from the definition of quantum conditional entropy (specifically,~\eqref{eq:cond-maj-def-1}--\eqref{eq:cond-maj-def-3} and~\eqref{eq:QCE-1st-postulate-mono}) that it is invariant under local isometric channels acting on system~$A$. So we prove invariance under isometric channels acting on system $B$. The inequality 
\begin{equation}
\H(A|B)_{\rho} \leq \H(A|B')_{\omega}    
\end{equation}
follows because $\mathcal{V}_{B\to B'}$ is a conditionally unital channel. Let $\tau_B$ be a state. The reverse inequality follows because 
\begin{equation}
\mathcal{R}^{\mathcal{V}}_{B'\to B}(\cdot) \coloneqq  \mathcal{V}^{\dag}(\cdot) + \operatorname{Tr}[(\operatorname{id}_{B'} - \mathcal{V}^\dag)(\cdot)]\tau_B    
\end{equation}
is a channel that is a left inverse of $\mathcal{V}_{B\to B'}$~\cite[Section~4.6.3]{W17}. Thus, $\mathcal{R}^{\mathcal{V}}_{B'\to B}$ is a conditionally unital channel, from which we conclude that
\begin{equation}
\H(A|B')_{\omega} \leq \H(A|B)_{\mathcal{R}^{\mathcal{V}}(\omega)} = \H(A|B)_{\rho}.    
\end{equation}
This concludes the proof.

\section{A subclass of conditional entropies}

\label{app:alt-cond-ent}

Here we consider a subclass of conditional entropies in which the postulate in~\eqref{eq:QCE-1st-postulate-mono} is replaced by monotonicity under conditionally unital channels. Since the set of conditionally unital channels contains the set of conditionally mixing channels, every function that is a conditional entropy according to this alternative definition is also a conditional entropy $\H$, as defined in~\eqref{eq:cond-ent-gen-def-1st}--\eqref{eq:QCE-2nd-postulate-add}.
We will later see that this approach is consistent and sufficient for deriving various properties of conditional entropy.

Let us begin by defining another version of a conditional majorization pre-order on bipartite quantum states, based on conditionally unital channels exclusively. We call this approach relaxed conditional majorization.

\begin{definition}[Relaxed conditional majorization]
Given two bipartite states $\rho_{AB}$ and $\sigma _{A'B'}$, we say that $\rho_{AB}$ conditionally majorizes $\sigma _{A'B'}$ in the relaxed sense and with respect to system $A$, and write 
\begin{equation}
    \rho_{AB}\succ ^{ A}\sigma_{A'B'}
\end{equation}
if there exist a system $A''$ such that $|A''| \geq \max\{|A|,|A'|\}$, isometric channels $\mathcal{V}_{A\to A''}$ and $\mathcal{U}_{A'\to A''}$,  and a conditionally unital channel~$\mathcal{N}_{A''B \rightarrow A''B'}$  such that
\begin{equation}
    \mathcal{U}_{A'\to A''}(\sigma _{A'B'}) =(\mathcal{N}_{A''B\to A''B'} \circ \mathcal{V}_{A\to A''})(\rho_{AB}) \;.
 \end{equation}
\end{definition}

We define a function $f$ to be a conditional uncertainty measure in the relaxed sense if it is not equal to the zero function and if it reverses the relaxed conditional majorization pre-order, i.e., 
\begin{equation}
\rho_{AB} \succ^A \sigma_{A'B'} \quad \Rightarrow \quad f(\rho_{AB}) \leq f(\sigma_{A'B'}).
\end{equation}

We now define this other notion of conditional entropy. Formally, a function
\begin{equation}
\H^{\blacktriangle} : \bigcup_{A, B} \mathfrak{D}(AB) \rightarrow \mathbb{R}
\label{eq:alt-cond-ent-gen-def-1st}
\end{equation}
is a quantum conditional entropy if it is not equal to the zero function and satisfies the following postulates for all systems $A$, $B$, $A'$, $B'$ and states $\rho_{AB}$ and $\sigma_{A'B'}$:
\begin{enumerate}
    \item Monotonicity: 
    \begin{equation}
      \rho_{AB}\succ ^{ A}\sigma_{A'B'} \  \Rightarrow \  \H^{\blacktriangle}(A|B)_{\rho} \leq \H^{\blacktriangle}(A'|B')_{\sigma}. 
    \label{eq:alt-QCE-1st-postulate-mono}   
    \end{equation}
    
    \item Additivity:
    \begin{equation}
    \H^{\blacktriangle}(AA'|BB')_{\rho \otimes \sigma} = \H^{\blacktriangle}(A|B)_{\rho} + \H^{\blacktriangle}(A'|B')_{\sigma}    .
        \label{eq:alt-QCE-2nd-postulate-add}
    \end{equation}
    
\end{enumerate}

Let us now recall a different definition of conditional min-entropy from~\cite{Renner2005}
\begin{equation}
H_{\min}^{\uparrow}(A|B)_{\rho}  \coloneqq \sup_{\sigma_B \in \mathfrak{D}(B)} H_{\min}(\rho_{AB}|\sigma_{B}),
\label{eq:alt-cond-min-ent-def}
\end{equation}
where
\begin{equation}
H_{\min}(\rho_{AB}|\sigma_{B}) \coloneqq - \inf_{\lambda \geq 0} \log_2 \{\lambda : \rho_{AB} \leq \lambda I_A \otimes \sigma_B\} .
\end{equation}
Note that $H_{\min}^{\uparrow}(A|B)_{\rho}$ is a conditional entropy $\H^{\blacktriangle}$ because it is additive~\cite[page~4342]{KRS09}, and~\cite[Lemma~3.1.12]{Renner2005} implies that $H_{\min}^{\uparrow}(A|B)_{\rho}$ is non-decreasing under conditionally unital channels. 

We then have the following findings, analogous to those from Theorem~\ref{thm:neg-cond-ent} and Corollary~\ref{cor:pos-QCE-reduct}, respectively.

\begin{theorem}
\label{thm:alt-neg-cond-ent}
Let $\H^{\blacktriangle}$ be a quantum conditional entropy, as defined in~\eqref{eq:alt-cond-ent-gen-def-1st}--\eqref{eq:alt-QCE-2nd-postulate-add}. Then, for every state $\rho_{AB}$,
\begin{equation}
\H^{\blacktriangle}(A|B)_\rho\geq H^{\uparrow}_{\min}(A|B)_\rho.
\label{eq:alt-main-thm}
\end{equation}
Equality is attained in~\eqref{eq:alt-main-thm} if $\rho_{AB}$ is equal to a maximally entangled state $\Phi^{(k)}_{A'B'}$ by the action of local isometric channels. Thus,
\begin{equation}
\H^{\blacktriangle}(A|B)_{\Phi^{(k)}} = -\log_2 k
\label{eq:alt-main-thm-equality}
\end{equation}
for every conditional entropy $\H^{\blacktriangle}$.
\end{theorem}

\begin{corollary}
\label{cor:alt-pos-QCE-reduct}
Let $\rho_{AB}$ be a bipartite state. Then the following are equivalent:
\begin{enumerate}
\item For every conditional entropy $\H^{\blacktriangle}$,
\begin{equation}
\H^{\blacktriangle}(A|B)_\rho\geq 0.
\end{equation}
\item There exists a state $\sigma_B$ such that $I_A\otimes\sigma_B\geq\rho_{AB}$.
\end{enumerate}
\end{corollary}

\begin{proof}
The proof of Corollary~\ref{cor:alt-pos-QCE-reduct} is immediate from Theorem~\ref{thm:alt-neg-cond-ent}. Indeed, if the first condition holds, then the inequality $H^{\uparrow}_{\min}(A|B)_\rho\geq 0$ holds in particular (since $H^{\uparrow}_{\min}(A|B)_\rho$ is a conditional entropy), and the second condition is then a consequence of the definition in~\eqref{eq:alt-cond-min-ent-def}. Now suppose that the second condition holds. This then means that there exists $\lambda \leq 1$ and a state $\sigma_B$ such that $\rho_{AB}\leq \lambda I_A\otimes\sigma_B$. According to the definition in~\eqref{eq:alt-cond-min-ent-def}, it follows that $H^{\uparrow}_{\min}(A|B)_\rho\geq 0$, and then Theorem~\ref{thm:alt-neg-cond-ent} implies the first condition.
\end{proof}

\bigskip
We have opted to emphasize the conditional entropy~$\H$ over $\H^{\blacktriangle}$ in the main text because
\begin{enumerate}
\item conditionally mixing channels have a strong motivation in terms of no information leakage from Alice to Bob, as would be expected for a measure of conditional uncertainty;
\item conditional majorization defined in this way reduces to classical conditional majorization, as defined in~\cite{GGHKLV18}, for the classical case;
\item the conditional min-entropy in~\eqref{eq:cond-min-ent-def} is in fact the smallest conditional entropy among all conditional entropies consistent with both sets of postulates;
\item every conditional entropy defined in this other way is in fact a conditional entropy $\H$, because every function that is monotone under all conditionally unital channels is also monotone under the subset of all conditionally mixing channels.
\end{enumerate}

\subsection{Proof of Theorem~\ref{thm:alt-neg-cond-ent}}

\label{app:proof-main-thm-alt-CE}

Let $\sigma_{B}$ be a state that achieves the optimal value of $H_{\min}^{\uparrow}(A|B)_{\rho}$, in the definition in~\eqref{eq:alt-cond-min-ent-def}.

Let us first consider the case when $H_{\min}^{\uparrow}(A|B)_{\rho}\geq0$. Set%
\begin{align}
k  & \coloneqq \left\vert A\right\vert , \label{eq:alt-cond-ent-main-thm-proof-1-1}\\
m  & \coloneqq \left\lfloor 2^{H_{\min}^{\uparrow}(A|B)_{\rho}}\right\rfloor
,\label{eq:m-def-non-neg-CE}%
\end{align}
and let $X$ be a classical system with dimension $k$. By the assumption that
$H_{\min}^{\uparrow}(A|B)_{\rho}\geq0$, and given the dimension bound $H_{\min
}^{\uparrow}(A|B)_{\rho}\leq\log_{2}k$~\cite[Lemma~5.2]{KR11}, it follows that $m\in\left[  k\right]  $.

Without loss of generality, we can suppose that $k>m$; if it is not the case,
then we can embed the system $A$ of $\rho_{AB}$ in a larger Hilbert space, so
that $H_{\min}^{\uparrow}(A|B)_{\rho}$ stays constant but $k$ increases. Set $\Pi_{X}$ to
be a projection onto an $m$-dimensional subspace of $X$ (i.e., set $\Pi
_{X}\coloneqq \sum_{x=1}^{m}|x\rangle\!\langle x|_{X}$, where $\{|x\rangle\}_{x}$ is an
orthonormal basis specifying the classical system $X$), and define $\tau
_{AB}^{(m)}$ as follows:%
\begin{equation}
\tau_{AB}^{(m)}\coloneqq \frac{I_{A}\otimes\sigma_{B}-m\rho_{AB}}{k-m}.
\end{equation}
The operator $\tau_{AB}^{(m)}$ is positive semi-definite because $k-m>0$ and%
\begin{align}
& \frac{1}{m}I_{A}\otimes\sigma_{B}-\rho_{AB}\nonumber\\
& \geq2^{-H_{\min}^{\uparrow}(A|B)_{\rho}}I_{A}\otimes\sigma_{B}-\rho_{AB}\\
& =2^{D_{\max}(\rho_{AB}\Vert I_{A}\otimes\sigma_{B})}I_{A}\otimes\sigma
_{B}-\rho_{AB}\\
& \geq0.
\end{align}
Also, $\tau_{AB}^{(m)}$ has trace equal to one and is thus a state. Let us now
define the following measure-and-prepare channel:
\begin{multline}
\mathcal{N}_{X\rightarrow AB}(\omega_{X})\coloneqq \\
\operatorname{Tr}[\Pi_{X}\omega
_{X}]\rho_{AB}
+\operatorname{Tr}[\left(  I_{X}-\Pi_{X}\right)  \omega_{X}
]\tau_{AB}^{(m)}.
\label{eq:ch-def-1}
\end{multline}
Indeed, the action of this channel is to perform a measurement according to
the POVM\ $\left\{  \Pi_{X},I_{X}-\Pi_{X}\right\}  $ and prepare the state
$\rho_{AB}$ if the first outcome is obtained and the state $\tau_{AB}^{(m)}$
if the second outcome is obtained. It should be understood that Alice has
access to the input system $X$.

This channel is conditionally unital because%
\begin{align}
& \mathcal{N}_{X\rightarrow AB}(I_{X})  \notag \\
& =\operatorname{Tr}[\Pi_{X}I_{X}%
]\rho_{AB}+\operatorname{Tr}[\left(  I_{X}-\Pi_{X}\right)  I_{X}]\tau
_{AB}^{(m)} \label{eq:cond-unit-1-1} \\
& =m\rho_{AB}+\left(  k-m\right)  \tau_{AB}^{(m)}\\
& =I_{A}\otimes\sigma_{B}.
\label{eq:cond-unit-1-3}
\end{align}
Also, if we input the uniform state $\mathbf{u}^{\Pi}_{X}\coloneqq \Pi_{X}/m$ on the
subspace onto which $\Pi_{X}$ projects, the output is given by
\begin{equation}
\mathcal{N}_{X\rightarrow AB}(\mathbf{u}^{\Pi}_{X})=\rho_{AB}
.\label{eq:ch-identity-hmin-lower->0}
\end{equation}
Putting these observations together, we conclude that
\begin{align}
\log_{2}\!\left\lfloor 2^{H_{\min}^{\uparrow}(A|B)_{\rho}}\right\rfloor  & =\mathbf{H}^{\blacktriangle}
(X)_{\mathbf{u}^{\Pi}} \label{eq:hmin-<=-cond-ent-+case-1}\\
& \leq\mathbf{H}^{\blacktriangle}(A|B)_{\mathcal{N}(\mathbf{u}^{\Pi})}\\
& =\mathbf{H}^{\blacktriangle}(A|B)_{\rho}.
\label{eq:hmin-<=-cond-ent-+case-3}
\end{align}
The first equality follows because the entropy of a uniform state is equal to
the logarithm of its rank and from the definition of $m$ in
\eqref{eq:m-def-non-neg-CE}. The inequality follows from the fact that
$\mathcal{N}_{X\rightarrow AB}$ is conditionally unital (see~\eqref{eq:cond-unit-1-1}--\eqref{eq:cond-unit-1-3}) and $\mathbf{H}^{\blacktriangle}$ is a
quantum conditional entropy, thus obeying the first postulate in~\eqref{eq:alt-QCE-1st-postulate-mono}. The final equality follows from~\eqref{eq:ch-identity-hmin-lower->0}.

Therefore, since all conditional entropies are additive for tensor-product states, we conclude that
\begin{align}
\H^{\blacktriangle}(A|B)_{\rho}&=\lim_{n\to\infty}\frac1n\H^{\blacktriangle}(A^n|B^n)_{\rho^{\otimes n}}\\
&\geq \lim_{n\to\infty}\frac1n\log_2\!\left\lfloor2^{H_{\min}^{\uparrow}(A^n|B^n)_{\rho^{\otimes n}}}\right\rfloor\\
&=\lim_{n\to\infty}\frac1n\log_2\!\left\lfloor2^{nH_{\min}^{\uparrow}(A|B)_{\rho}}\right\rfloor\\
&=H_{\min}^\ua(A|B)_\rho\;.
\label{eq:alt-cond-ent-main-thm-proof-1-last}
\end{align}
The inequality follows from~\eqref{eq:hmin-<=-cond-ent-+case-1}--\eqref{eq:hmin-<=-cond-ent-+case-3}. The second equality follows from additivity, and the final equality from the assumption that $H_{\min}^{\uparrow}(A|B)_\rho\geq 0$.

Now let us consider the case when $H_{\min}^{\uparrow}(A|B)<0$. Let%
\begin{align}
t  & \coloneqq 2^{-H_{\min}^{\uparrow}(A|B)},
\label{eq:alt-cond-ent-main-thm-proof-2-1}\\
k^{\prime}  & \coloneqq \left\lceil t\right\rceil ,\\
k  & \coloneqq \left\vert A\right\vert .
\end{align}
Observe that $t>1$ by assumption. Also, observe that $t$ is the smallest real
satisfying $\rho_{AB}\leq tI_{A}\otimes\sigma_{B}$. Let $A^{\prime}$ be a
system of dimension $k^{\prime}$, and let $X$ be a classical system of
dimension $\left\vert X\right\vert \coloneqq kk^{\prime}$. Now define%
\begin{equation}
\tau_{AB}^{(k,k^{\prime})}\coloneqq \frac{k^{\prime}I_{A}\otimes\sigma_{B}-\rho_{AB}%
}{kk^{\prime}-1},
\end{equation}
and observe that $\tau_{AB}^{(k,k^{\prime})}$ is positive semi-definite
because $kk^{\prime}>1$ and%
\begin{align}
& k^{\prime}I_{A}\otimes\sigma_{B}-\rho_{AB} \notag \\
& \geq2^{-H_{\min}^{\uparrow}(A|B)}I_{A}\otimes\sigma_{B}-\rho_{AB}\\
& =2^{D_{\max}(\rho_{AB}\Vert I_{A}\otimes\sigma_{B})}I_{A}\otimes\sigma
_{B}-\rho_{AB}\\
& \geq0.
\end{align}
Since $\tau_{AB}^{(k,k^{\prime})}$ has trace equal to one, it follows that it
is a state. We then construct the following channel:%
\begin{multline}
\mathcal{N}_{X\rightarrow AA^{\prime}B}(\omega_{X})\coloneqq \\
\left(  \operatorname{Tr}%
[|1\rangle\!\langle1|_{X}\omega_{X}]\rho_{AB}+\operatorname{Tr}[\Pi_{X}%
\omega_{X}]\tau_{AB}^{(k,k^{\prime})}\right)  \otimes\mathbf{u}_{A^{\prime}},
\label{eq:ch-def-2}
\end{multline}
where
\begin{equation}
\Pi_{X}\coloneqq I_{X}-|1\rangle\!\langle1|_{X},
\end{equation}
so that $\operatorname{Tr}%
[\Pi_{X}]=kk^{\prime}-1$. Also, it should be understood that Alice has access
to the $X$ input system. Consider that%
\begin{equation}
\mathcal{N}_{X\rightarrow AA^{\prime}B}(|1\rangle\!\langle1|_{X})=\rho
_{AB}\otimes\mathbf{u}_{A^{\prime}}.\label{eq:cond-ent-neg-key-identity}%
\end{equation}
Furthermore, the channel is conditionally unital because%
\begin{align}
& \mathcal{N}_{X\rightarrow AA^{\prime}B}(I_{X})  \notag \\
& =\left(  \operatorname{Tr}%
[|1\rangle\!\langle1|_{X}I_{X}]\rho_{AB}+\operatorname{Tr}[\Pi_{X}I_{X}%
]\tau_{AB}^{(k,k^{\prime})}\right)  \otimes\mathbf{u}_{A^{\prime}}\notag \\
& =\left(  \rho_{AB}+\left(  kk^{\prime}-1\right)  \tau_{AB}^{(k,k^{\prime}%
)}\right)  \otimes\mathbf{u}_{A^{\prime}}\\
& =k^{\prime}I_{A}\otimes\sigma_{B}\otimes\mathbf{u}_{A^{\prime}}\\
& =I_{A}\otimes\sigma_{B}\otimes I_{A^{\prime}}.
\end{align}
Then it follows that%
\begin{align}
0  & =\mathbf{H}^{\blacktriangle}(X)_{|1\rangle\!\langle1|} \label{eq:neg-min-ent-ineq-1} \\
& \leq\mathbf{H}^{\blacktriangle}(AA^{\prime}|B)_{\mathcal{N}(|1\rangle\!\langle1|)}\\
& =\mathbf{H}^{\blacktriangle}(AA^{\prime}|B)_{\rho\otimes\mathbf{u}}\\
& =\mathbf{H}^{\blacktriangle}(A|B)_{\rho}+\mathbf{H}(A^{\prime})_{\mathbf{u}}\\
& =\mathbf{H}^{\blacktriangle}(A|B)_{\rho}+\log_{2}\left\lceil 2^{-H_{\min}^{\uparrow}(A|B)}\right\rceil .
\label{eq:neg-min-ent-ineq-last}
\end{align}
The first equality follows because the entropy of a pure state is equal to
zero. The inequality follows because $\mathcal{N}_{X\rightarrow AA^{\prime}B}$
is a conditionally unital channel and $\mathbf{H}^{\blacktriangle}$ is a
quantum conditional entropy, thus obeying the first postulate in~\eqref{eq:alt-QCE-1st-postulate-mono}. The second equality follows from
\eqref{eq:cond-ent-neg-key-identity}. The penultimate equality follows from
the additivity postulate of quantum conditional entropy. The final equality
follows because the entropy of the uniform state $\mathbf{u}_{A^{\prime}}$ is
equal to the logarithm of its rank, which is in turn equal to $\left\lceil
2^{-H_{\min}^{\uparrow}(A|B)}\right\rceil $. 

Now, since $\H^{\blacktriangle}(A|B)_\rho$ is additive for tensor-product states, we find that
\begin{align}
& \H^{\blacktriangle}(A|B)_{\rho}\notag \\
&=\lim_{n\to\infty}\frac1n\H^{\blacktriangle}(A^n|B^n)_{\rho^{\otimes n}}\\
&\geq \lim_{n\to\infty}-\frac1n\log_2\!\left\lceil2^{-H_{\min}^{\uparrow}(A^n|B^n)_{\rho^{\otimes n}}}\right\rceil\\
&=\lim_{n\to\infty}-\frac1n\log_2\!\left\lceil2^{-nH_{\min}^{\uparrow}(A|B)_{\rho}}\right\rceil\\
&=\lim_{n\to\infty}-\frac1n\log_2\!\left\lceil t^n\right\rceil\\
&\geq\lim_{n\to\infty}-\frac1n\log_2\!\left(t^n+1\right)\\
&=-\log_2 t\\
& =H_{\min}^{\uparrow}(A|B)_\rho\;.
\label{eq:alt-cond-ent-main-thm-proof-2-last}
\end{align}
The first inequality follows from~\eqref{eq:neg-min-ent-ineq-1}--\eqref{eq:neg-min-ent-ineq-last}. The second equality follows from additivity, and the penultimate equality from the assumption that $t>1$.
This completes the proof of the lower bound. 

It is left to prove the equality on maximally entangled states, i.e.,~\eqref{eq:alt-main-thm-equality}.
Since conditional entropy is invariant under local isometries, we can assume without loss of generality  that $k \eqdef |A| = |B|$. Now, consider the state $\Phi^{(k)}_{AB} \otimes \mathbf{u}_{\tA}$ where $|\tA| = k$, and observe that by introducing a trivial system $B_{2}$ such that $|B_{2}|=1$ we get
\begin{align}
     & \H^{\blacktriangle}(A\tA|B)_{\Phi^{(k)}_{AB} \otimes \mathbf{u}_{\tA}} \notag \\
    &=   \H^{\blacktriangle}(A\tA|B B_{2})_{\Phi^{(k)}_{AB} \otimes \mathbf{u} _{\tA B_{2}}} \label{eq:additive-max-ent-unif-1}\\
    & = \H^{\blacktriangle}(A|B)_{\Phi^{(k)}} + \H^{\blacktriangle}(\tA | B_{2})_{\mathbf{u} _{\tA B_{2}}} \\
    & = \H^{\blacktriangle}(A|B)_{\Phi^{(k)}} + \H^{\blacktriangle}(\mathbf{u}_{\tA}) \\
    &= \H^{\blacktriangle}(A|B)_{\Phi^{(k)}} + \log_2 k,
    \label{eq:additive-max-ent-unif-last}
\end{align}
where we used the additivity property of conditional entropy and the fact that the entropy of the uniform state $\mathbf{u}_{\tA}$ is equal to $\log_2 k$. It is therefore sufficient to show that
\begin{equation}
\H^{\blacktriangle}(A\tA|B)_{\Phi^{(k)}_{AB} \otimes \mathbf{u}_{\tA}} \leq 0 ,   
\end{equation}
because the above equation implies that 
\be
\H^{\blacktriangle}(A|B)_{\Phi^{(k)}}\leq-\log_2 k=H_{\min}^{\uparrow}(A|B)_{\Phi^{(k)}}\;.
\ee
For this purpose, let $X$ be a classical system of dimension $\left\vert
X\right\vert =k^{2}$ and let $\mathcal{N}_{A\tilde{A}B\rightarrow X}$ be a
quantum channel defined by%
\begin{equation}
\mathcal{N}_{A\tilde{A}B\rightarrow X}   \coloneqq \mathcal{E}_{AB\rightarrow X}%
\circ\operatorname{Tr}_{\tilde{A}},
\label{eq:channel-3-pf}
\end{equation}
where
\begin{multline}
\mathcal{E}_{AB\rightarrow X}(\omega_{AB})   \coloneqq \operatorname{Tr}[\Phi
_{AB}^{k}\omega_{AB}]|1\rangle\!\langle1|_{X}\\
+\operatorname{Tr}[(I_{AB}-\Phi_{AB}^{(k)})\omega_{AB}]\frac
{I_{X}-|1\rangle\!\langle1|_{X}}{k^{2}-1}.
\end{multline}
This channel is conditionally unital because%
\begin{align}
& \mathcal{N}_{A\tilde{A}B\rightarrow X}(I_{A}\otimes I_{\tilde{A}}%
\otimes\sigma_{B})\nonumber\\
& =(\mathcal{E}_{AB\rightarrow X}\circ\operatorname{Tr}_{\tilde{A}}%
)(I_{A}\otimes I_{\tilde{A}}\otimes\sigma_{B})\\
& =k \mathcal{E}_{AB\rightarrow X}(I_{A}\otimes\sigma_{B})\\
& =k \left(
\begin{array}
[c]{c}%
\operatorname{Tr}[\Phi_{AB}^{(k)}(I_{A}\otimes\sigma_{B})]|1\rangle\!\langle
1|_{X}\\
+\operatorname{Tr}[(I_{AB}-\Phi_{AB}^{(k)})(I_{A}\otimes\sigma_{B})]\frac
{I_{X}-|1\rangle\!\langle1|_{X}}{k^{2}-1}%
\end{array}
\right)  \\
& =k \left(  \frac{1}{k}|1\rangle\!\langle1|_{X}+\left(  k-\frac{1}{k}\right)
\frac{I_{X}-|1\rangle\!\langle1|_{X}}{k^{2}-1}\right)  \\
& =|1\rangle\!\langle1|_{X}+\left(  k^{2}-1\right)  \frac{I_{X}-|1\rangle
\langle1|_{X}}{k^{2}-1}\\
& =I_{X}.
\end{align}
Also, consider that%
\begin{equation}
\mathcal{N}_{A\tilde{A}B\rightarrow X}(\Phi_{AB}^{(k)}\otimes\mathbf{u}%
_{\tilde{A}})=|1\rangle\!\langle1|_{X}%
.\label{eq:max-ent-state-identity-cond-unit}%
\end{equation}
Then%
\begin{align}
\H^{\blacktriangle}(A|B)_{\Phi^{(k)}} + \log_2 k & = \mathbf{H}^{\blacktriangle}(A\tilde{A}|B)_{\Phi^{(k)}\otimes\mathbf{u}}  \\
& \leq\mathbf{H}^{\blacktriangle}%
(X)_{\mathcal{N}(\Phi^{(k)}\otimes\mathbf{u})}\\
& =\mathbf{H}^{\blacktriangle}(X)_{|1\rangle\!\langle1|}\\
& =0.
\label{eq:final-final-proof-end}
\end{align}
The first equality follows from~\eqref{eq:additive-max-ent-unif-1}--\eqref{eq:additive-max-ent-unif-last}. The inequality follows because the channel $\mathcal{N}_{A\tilde
{A}B\rightarrow X}$ is conditionally unital. The second equality follows from
\eqref{eq:max-ent-state-identity-cond-unit} and the last because the entropy
of every pure state is equal to zero. This concludes the proof.

\subsection{Proof of Theorem~\ref{thm:neg-cond-ent}}

\label{app:proof-main-thm-main-txt}

Here the idea is essentially the same as that given in Appendix~\ref{app:proof-main-thm-alt-CE}, except that we need to construct channels that are conditionally mixing (i.e., both conditionally unital and $A \not \to B$ semi-causal). 

First, throughout, we make the following replacement in the proof:
\begin{equation}
H_{\min}^{\uparrow}(A|B)_{\rho} \quad \to \quad 
H_{\min}(A|B)_{\rho}. 
\label{eq:proof-replacement}
\end{equation}

To handle the case when $H_{\min}(A|B)_{\rho} \geq 0$, we employ the same channel in~\eqref{eq:ch-def-1}, with the exception that we set
\begin{equation}
\tau_{AB}^{(m)}\coloneqq \frac{I_{A}\otimes\rho_{B}-m\rho_{AB}}{k-m}.
\label{eq:tau-replacement-1}
\end{equation}
All of the steps in~\eqref{eq:alt-cond-ent-main-thm-proof-1-1}--\eqref{eq:alt-cond-ent-main-thm-proof-1-last} go through with the replacements in~\eqref{eq:proof-replacement} and~\eqref{eq:tau-replacement-1}. We now verify that the $A\not \to B$ semi-causal constraint holds for this channel. Consider, for the modified channel, that
\begin{align}
& \operatorname{Tr}_{A}[\mathcal{N}_{X\to AB}(\omega_X)] \notag \\
& = \operatorname{Tr}[\Pi_{X}\omega
_{X}]\operatorname{Tr}_{A}[\rho_{AB}] \notag \\
& \qquad +\operatorname{Tr}[\left(  I_{X}-\Pi_{X}\right)  \omega_{X}%
]\operatorname{Tr}_{A}[\tau_{AB}^{(m)}]  \\
&  = \operatorname{Tr}[\Pi_{X}\omega
_{X}]\rho_{B} +\operatorname{Tr}[\left(  I_{X}-\Pi_{X}\right)  \omega_{X}%
]\rho_B  \\
&  = \rho_B.
\end{align}
Thus, the output state of $B$ has no dependence on the input state $\omega_X$ when system $A$ is traced out, so that the semi-causal constraint holds.

Now let us consider the case when $H_{\min}(A|B)_{\rho} < 0$. We employ the same channel in~\eqref{eq:ch-def-2}, with the exception that we set
\begin{equation}
\tau_{AB}^{(k,k^{\prime})}\coloneqq \frac{k^{\prime}I_{A}\otimes\rho_{B}-\rho_{AB}%
}{kk^{\prime}-1}.
\label{eq:tau-replacement-2}
\end{equation}
Then all of the steps in~\eqref{eq:alt-cond-ent-main-thm-proof-2-1}--\eqref{eq:alt-cond-ent-main-thm-proof-2-last} go through, with the replacements in~\eqref{eq:proof-replacement} and~\eqref{eq:tau-replacement-2}. We now verify that the $A \not \to B$ semi-causal constraint holds for this channel. Consider, for the modified channel, that
\begin{align}
& \operatorname{Tr}_{AA'}[\mathcal{N}_{X\rightarrow AA^{\prime}B}(\omega_{X})] \notag \\
& = \Big(  \operatorname{Tr}%
[|1\rangle\!\langle1|_{X}\omega_{X}]\operatorname{Tr}_{A}[\rho_{AB}]\notag \\
& \qquad +\operatorname{Tr}[\Pi_{X}%
\omega_{X}]
\operatorname{Tr}_{A}[\tau_{AB}^{(k,k^{\prime})}]\Big)   \otimes\operatorname{Tr}_{A'}[\mathbf{u}_{A^{\prime}}]\\
& =   \operatorname{Tr}%
[|1\rangle\!\langle1|_{X}\omega_{X}]\rho_{B}+\operatorname{Tr}[\Pi_{X}%
\omega_{X}]
\rho_{B} \\
& = \rho_B.
\end{align}
Thus, the output state of $B$ has no dependence on the input state $\omega_X$ when system $A$ is traced out, so that the semi-causal constraint holds. This establishes the inequality in~\eqref{eq:alt-main-thm}.

Finally, we consider the equality in~\eqref{eq:alt-main-thm-equality}. The argument here is actually exactly the same as in~\eqref{eq:additive-max-ent-unif-1}--\eqref{eq:final-final-proof-end}, except using the replacement in~\eqref{eq:proof-replacement}. We just need to verify that the channel in~\eqref{eq:channel-3-pf} is $A \not \to B$ semi-causal. However, the only output is system $X$, which belongs to Alice. Thus, after tracing it out, there is no remaining system, and the channel is then semi-causal.

\section{Proof of Theorem~\ref{thm:cond-ent-from-rel-ent}}

\label{app:proof-cond-ent-from-rel-ent}

First, we demonstrate that $\H$ satisfies the monotonicity property of conditional entropy from Definition~\ref{def:cond-maj-cond-mix}. Let $\mathcal{N}_{AB \rightarrow A'B'}$ be a conditionally mixing channel, and consider the bipartite density matrix $\rho_{AB}$. We begin by considering the case where $A = A'$ and
\begin{multline}
    \H(A \big| B)_{\mathcal{N}(\rho)} = 
    \\
    \log_2 |A| - \D(\mathcal{N}(\rho_{AB}) \big\| \mathbf{u}_{A} \otimes \operatorname{Tr}_{A}\!\left[\mathcal{N}(\rho_{AB})\right]) 
\end{multline}
and from using the fact that $\mathcal{N}$ is semi-causal and conditionally unital, we find that
\begin{align}
    & \mathbf{u}_{A} \otimes \operatorname{Tr}_{A}\!\left[\mathcal{N}(\rho_{AB}) \right] \notag \\
    &= \mathbf{u}_{A} \otimes \operatorname{Tr}_{A}\!\left[\mathcal{N}\left(\mathbf{u}_{A} \otimes \rho_{B}\right)\right] \\
    &= \mathcal{N}(\mathbf{u}_{A} \otimes \rho_{B}),
\end{align}
from which we have 
\begin{align}
    & \H(A \big| B)_{\mathcal{N}(\rho)} \notag \\
    &= \log_2 |A| - \D\!\left(\mathcal{N}(\rho_{AB}) \big\| \mathcal{N}(\mathbf{u}_{A} \otimes \rho_{B}) \right)  \\
    &\geq \log_2 |A| - \D\!\left(\rho_{AB} \big\| \mathbf{u}_{A} \otimes \rho_{B} \right) ,
\end{align}
where the last line follows from the data-processing inequality.

We also need to prove that it is invariant under the action of a local isometric channel acting on system~$A$. Let $\mathcal{V}_{A \rightarrow A'}$ be an arbitrary  isometric channel, and let~$C$ represent a Hilbert space such that $|C| = |A'| - |A|$. Then 
\begin{multline}
    \H(A | B )_{\mathcal{V}(\rho)} = \\
    \log_2 |A'| - \D(\mathcal{V}_{A \rightarrow A'}(\rho_{AB}) \| \mathbf{u}_{A'} \otimes \rho_{B}).
\end{multline}
Clearly, if $\mathcal{V}$ is a unitary channel such that $A \simeq  A'$, then $\H(A' | B)_{\mathcal{V}(\rho)} = \H(A | B)_{\rho}$. Then without loss of generality, we assume that 
\begin{align}
    \mathcal{V}_{A \rightarrow A'}(\rho_{AB}) & = \rho_{AB} \oplus \mathbf{0}_{CB} \\
&    = \begin{pmatrix}
    \rho_{AB} & \mathbf{0} \\
    \mathbf{0} & \mathbf{0}_{CB}
    \end{pmatrix},
\end{align}
because the conditional entropy is invariant under local unitaries. We can additionally rewrite
\begin{equation}
\mathbf{u}_{A'} = t \mathbf{u}_{A} \oplus (1-t) \mathbf{u}_{C},
\end{equation}
where $t \coloneqq  \frac{|A|}{|A'|}$. From this, it follows that 
\begin{align}
    & \D(\mathcal{V}_{A \rightarrow A'}(\rho_{AB})\big\| \mathbf{u}_{A'} \otimes \rho_{B}) \notag \\
    &= \D\!\left(\rho_{AB} \oplus \mathbf{0}_{CB} \big\| \left(t \mathbf{u}_{A} \otimes \rho_{B}\right)\oplus \left(\left(1-t\right)\mathbf{u}_{C} \otimes \rho_{B} \right)\right) \notag\\
    &= \D\!\left(\rho_{AB} \big\| t \mathbf{u}_{A} \otimes \rho_{B}\right) \notag\\
    &= \D\!\left(\rho_{AB}  \big\| \mathbf{u}_{A} \otimes \rho_{B}\right) - \log_2 t ,
\end{align}
where the second equality follows from a property of quantum relative entropy~\cite[Theorem~6.3.1]{Gour2023}.
Thus, upon substitution, we have that
\begin{align}
     & \H(A' | B)_{\mathcal{V}(\rho)} \notag \\
    &= \log_2 |A'| - \D\!\left(\rho_{AB} \big\| \mathbf{u}_{A} \otimes \rho_{B} \right) + \log_2 t \\
    &= \H(A | B)_{\rho}.
\end{align}

Putting the above two observations together, we now apply them to the definition given in Definition~\ref{def:cond-maj-cond-mix}. That is, suppose there exists a system $A''$ such that $|A''| \geq \max\{|A|,|A'|\}$, isometric channels $\mathcal{V}_{A\to A''}$ and $\mathcal{U}_{A'\to A''}$,  and a conditionally mixing channel~$\mathcal{N}_{A''B \rightarrow A''B'}$  such that
\begin{equation}
    \mathcal{U}_{A'\to A''}(\sigma _{A'B'}) =(\mathcal{N}_{A''B\to A''B'} \circ \mathcal{V}_{A\to A''})(\rho_{AB})
    \end{equation}
    Then
    \begin{align}
\H(A' | B')_{\sigma } & = \H(A'' | B')_{\mathcal{U}(\sigma )} \\
 & = \H(A'' | B')_{(\mathcal{N} \circ \mathcal{V})(\rho)} \\
 & \leq \H(A'' | B)_{\mathcal{V}(\rho)} \\
 & = \H(A | B)_{\rho}.
    \end{align}
    In the above steps, we employed the invariance of $\H$ under local isometric channels acting on the systems labeled by $A$, as well as the monotonicity under conditionally mixing channels.

Finally, additivity  of $\H$ follows immediately from additivity of~$\D$.


\subsection{Conditional min-entropy is a conditional entropy}

\label{app:cond-min-ent-is-cond-ent}

Here we justify that the conditional min-entropy is indeed a conditional entropy according to the definition in~\eqref{eq:cond-ent-gen-def-1st}--\eqref{eq:QCE-2nd-postulate-add}. Indeed, we use the definition of conditional min-entropy in~\eqref{eq:cond-min-ent-def} to rewrite it in terms of the max-relative entropy $D_{\max}(\rho\Vert \sigma)$ as
\begin{equation}
H_{\min}(A|B)_{\rho} = \log_2 |A| - D_{\max}(\rho_{AB} \Vert \u_A \otimes \rho_B),
\end{equation}
where the max-relative entropy of states $\rho$ and $\sigma$ is defined as~\cite{D09}:
\begin{equation}
D_{\max}(\rho\Vert \sigma) \coloneqq \log_2 \inf_{\lambda \geq 0} \{\rho 
\leq \lambda \sigma\}.
\end{equation}
Then we apply Theorem~\ref{thm:cond-ent-from-rel-ent}.

\section{Entropy inequalities}

\label{appG}

Entropy inequalities play a major role in information theory. For example, the strong subadditivity of the von Neumann entropy is a deep result in quantum Shannon theory that can be expressed as an inequality involving the von Neumann conditional entropy:
\be
H(A|B)_\rho\geq H(A|BC)_{\rho}\quad\forall\rho\in\md(ABC)\;,
\ee
where $H(A|B)_\rho\eqdef H(AB)_\rho-H(B)_\rho$ and $H(A)_\rho\eqdef-\tr[\rho_A\log_2\rho_A]$ is the von Neumann entropy. The above inequality holds for all conditional entropies: namely, for every conditional entropy $\H$, we have
\be
\H(A|B)_\rho\geq \H(A|BC)_{\rho}\quad\forall\rho\in\md(ABC)\;.
\ee
The above inequality is a simple consequence of the fact that tracing the subsystem $C$ is a conditionally mixing channel and therefore
\be
\rho_{ABC}\succ_A\rho_{AB}\;.
\ee

It is therefore natural to ask if there are other entropy inequalities that can be applied to all entropies. In this section we present a few such inequalities.

We show here that the teleportation protocol of communicating a quantum system from Bob to Alice is a conditionally mixing channel. This property implies in particular that (see arguments in Proposition~\ref{prop:other-ent-ineqs} below)
\be
\rho_{AB}\otimes\Phi_{\tA\tB}\succ_{\!\!A\tA}\rho_{A\tA}\;,
\ee
where $\Phi_{\tA\tB}$ is a maximally entangled state, $\rho_{A\tA}$ is the same state as $\rho_{AB}$ with system $B$ relabeled as $\tA$, and we assume for simplicity that all systems have the same dimension. The above relation immediately implies that every conditional entropy  $\H$ satisfies the following:
\begin{align}
\H(AB)_\rho & =\H(A\tA)_\rho\\
&\geq \H(A\tA|B\tB)_{\rho\otimes\Phi}\\
&=\H(A|B)_{\rho}+\H(\tA|\tB)_{\Phi}\;,
\end{align}
where the equality follows from additivity.
Now, from Theorem~\ref{thm:neg-cond-ent} we get that 
$\H(\tA|\tB)_{\Phi}=-\log_2|B|$, so that we arrive at the following entropy inequality:
\be
\H(AB)_\rho\geq \H(A|B)_{\rho}-\log_2|B|\;.
\ee
%

\begin{proposition}
\label{prop:other-ent-ineqs}
Let $\rho_{AB}$ be a bipartite state. Then%
\begin{align}
\mathbf{H}(A|B)_{\rho}  & \leq\mathbf{H}(AB)_{\rho}+\log_{2}\left\vert
B\right\vert ,\\
\mathbf{H}^{\blacktriangle}(A|B)_{\rho}  & \leq\mathbf{H}(AB)_{\rho}+\log
_{2}\left\vert B\right\vert .\label{eq:cond-ent-ineq-dim-bnd}%
\end{align}

\end{proposition}

\begin{proof}
Let $\Phi_{A^{\prime}B^{\prime}}$ be a maximally entangled state of Schmidt
rank $\left\vert B\right\vert $, with $\left\vert A^{\prime}\right\vert
=\left\vert B^{\prime}\right\vert =\left\vert B\right\vert $. We perform
teleportation to transmit system $A$ to $A^{\prime}$. In more detail, the
teleportation channel we consider is
\begin{multline}
\mathcal{N}_{A^{\prime}BB^{\prime}\rightarrow A^{\prime}}(\omega_{A^{\prime
}BB^{\prime}})\coloneqq \\
\sum_{x}\mathcal{U}_{A^{\prime}}^{x}\left(  \langle\Phi
^{x}|_{BB^{\prime}}\omega_{A^{\prime}BB^{\prime}}|\Phi^{x}\rangle_{BB^{\prime
}}\right)  ,
\end{multline}
where $\{|\Phi^{x}\rangle_{BB^{\prime}}\}_{x}$ is the standard Bell basis and
$\left\{  \mathcal{U}_{A^{\prime}}^{x}\right\}  _{x}$ is the set of unitary
correction channels in teleportation. The channel $\mathcal{N}_{A^{\prime
}BB^{\prime}\rightarrow A^{\prime}}$ is conditionally unital because%
\begin{align}
& \mathcal{N}_{A^{\prime}BB^{\prime}\rightarrow A^{\prime}}(I_{A^{\prime}%
}\otimes\sigma_{BB^{\prime}})\nonumber\\
& =\sum_{x}\mathcal{U}_{A^{\prime}}^{x}\left(  \langle\Phi^{x}|_{BB^{\prime}%
}\left(  I_{A^{\prime}}\otimes\sigma_{BB^{\prime}}\right)  |\Phi^{x}%
\rangle_{BB^{\prime}}\right)  \\
& =\sum_{x}\mathcal{U}_{A^{\prime}}^{x}(I_{A^{\prime}})\otimes\langle\Phi
^{x}|_{BB^{\prime}}\left(  \sigma_{BB^{\prime}}\right)  |\Phi^{x}%
\rangle_{BB^{\prime}}\\
& =\sum_{x}I_{A^{\prime}}\otimes\langle\Phi^{x}|_{BB^{\prime}}\left(
\sigma_{BB^{\prime}}\right)  |\Phi^{x}\rangle_{BB^{\prime}}\\
& =I_{A^{\prime}}\otimes\sum_{x}\langle\Phi^{x}|_{BB^{\prime}}\left(
\sigma_{BB^{\prime}}\right)  |\Phi^{x}\rangle_{BB^{\prime}}\\
& =I_{A^{\prime}}.
\end{align}
The channel $\mathcal{N}_{A^{\prime}BB^{\prime}\rightarrow A^{\prime}}$ is
also semi-causal because it is realized as a one-way LOCC\ channel from
$BB^{\prime}$ to $A^{\prime}$.\ Thus, $\mathcal{N}_{A^{\prime}BB^{\prime
}\rightarrow A^{\prime}}$ is conditionally mixing. Then consider that
\begin{align}
 & \mathbf{H}(A|B)_{\rho}-\log_{2}\left\vert B\right\vert 
\notag \\
& =\mathbf{H}(A|B)_{\rho}+\mathbf{H}(A^{\prime}|B^{\prime})_{\Phi}\\
& =\mathbf{H}(AA^{\prime}|BB^{\prime})\\
& \leq\mathbf{H}(AA^{\prime})_{\mathcal{N}(\rho\otimes\Phi)}\\
& =\mathbf{H}(AA^{\prime})_{\rho}\\
& =\mathbf{H}(AB)_{\rho}.
\end{align}
The first equality follows from Theorem~\ref{thm:neg-cond-ent}. The second equality follows
from the additivity postulate. The inequality follows because $\mathcal{N}%
_{A^{\prime}BB^{\prime}\rightarrow A^{\prime}}$ is conditionally mixing, as argued
above. The final equality is just a relabeling of systems.

The inequality in~\eqref{eq:cond-ent-ineq-dim-bnd} follows similarly,
using the fact that $\mathcal{N}_{A^{\prime}BB^{\prime}\rightarrow A^{\prime}%
}$ is conditionally unital.
\end{proof}

\begin{proposition}
Let $\rho_{AB}$ be a bipartite state. Then%
\begin{equation}
\mathbf{H}^{\blacktriangle}(A|B)_{\rho}\geq\mathbf{H}(AB)_{\rho}-\log
_{2}\left\vert B\right\vert .
\end{equation}

\end{proposition}

\begin{proof}
Consider the initial state%
\begin{equation}
\rho_{AA^{\prime}}\otimes|0\rangle\!\langle0|_{B^{\prime}},
\end{equation}
where $\left\vert A^{\prime}\right\vert =\left\vert B^{\prime}\right\vert
=\left\vert B\right\vert $ and we have relabeled $B$ of $\rho_{AB}$ as
$A^{\prime}$ (physically, the system $B$ begins on Alice's side, and so we
label it as $A^{\prime}$). Then perform the swap and prepare channel from
\cite[Definition~8]{Vempati2022unitaloperations}, which was shown there to be a conditionally
unital channel. That is, this channel is realized as $\mathcal{D}_{A^{\prime}%
}\circ\mathcal{F}_{A^{\prime}B^{\prime}}$, where $\mathcal{F}_{A^{\prime
}B^{\prime}}$ is the unitary swap channel and $\mathcal{D}_{A^{\prime}}$ is
the completely depolarizing channel (traces out its input and replaces with
the uniform state). This is conditionally unital because%
\begin{align}
(\mathcal{D}_{A^{\prime}}\circ\mathcal{F}_{A^{\prime}B^{\prime}}%
)(I_{A^{\prime}}\otimes\sigma_{B^{\prime}})  & =\mathcal{D}_{A^{\prime}%
}(\sigma_{A^{\prime}}\otimes I_{B^{\prime}})\\
& =\mathcal{D}_{A^{\prime}}(\sigma_{A^{\prime}})\otimes I_{B^{\prime}}\\
& =\mathbf{u}_{A^{\prime}}\otimes I_{B^{\prime}}\\
& =I_{A^{\prime}}\otimes\mathbf{u}_{B^{\prime}}.
\end{align}
However, this channel is not semi-causal because it allows for communication
from $A^{\prime}$ to $B^{\prime}$. That is,
\begin{align}
& (\mathcal{D}_{A^{\prime}}\circ\mathcal{F}_{A^{\prime}B^{\prime}}
)(\omega_{A^{\prime}}\otimes\sigma_{B^{\prime}}) \notag \\
& =\mathcal{D}_{A^{\prime}
}(\sigma_{A^{\prime}}\otimes\omega_{B^{\prime}})\\
& =\mathbf{u}_{A^{\prime}}\otimes\omega_{B^{\prime}}.
\end{align}
Now consider that
\begin{align}
& (\mathcal{D}_{A^{\prime}}\circ\mathcal{F}_{A^{\prime}B^{\prime}}%
)(\rho_{AA^{\prime}}\otimes|0\rangle\!\langle0|_{B^{\prime}})  \notag \\
& =\mathcal{D}%
_{A^{\prime}}(\rho_{AB^{\prime}}\otimes|0\rangle\!\langle0|_{A^{\prime}})\\
& =\rho_{AB^{\prime}}\otimes\mathcal{D}_{A^{\prime}}(|0\rangle\!\langle
0|_{A^{\prime}})\\
& =\rho_{AB^{\prime}}\otimes\mathbf{u}_{A^{\prime}}.
\end{align}
Then
\begin{align}
& \mathbf{H}(AB)_{\rho} \notag \\
& =\mathbf{H}(AA^{\prime})_{\rho}\\
& =\mathbf{H}^{\blacktriangle}(AA^{\prime}|B^{\prime})_{\rho_{AA^{\prime}%
}\otimes|0\rangle\!\langle0|_{B^{\prime}}}\\
& \leq\mathbf{H}^{\blacktriangle}(AA^{\prime}|B^{\prime})_{(\mathcal{D}%
_{A^{\prime}}\circ\mathcal{F}_{A^{\prime}B^{\prime}})(\rho_{AA^{\prime}%
}\otimes|0\rangle\!\langle0|_{B^{\prime}})}\\
& =\mathbf{H}^{\blacktriangle}(AA^{\prime}|B^{\prime})_{\rho_{AB^{\prime}%
}\otimes\mathbf{u}_{A^{\prime}}}\\
& =\mathbf{H}^{\blacktriangle}(A|B^{\prime})_{\rho_{AB^{\prime}}}%
+\mathbf{H}(A^{\prime})_{\mathbf{u}_{A^{\prime}}}\\
& =\mathbf{H}^{\blacktriangle}(A|B)_{\rho_{AB}}+\log_{2}\left\vert
B\right\vert .
\end{align}
The first equality follows from relabeling systems. The second equality
follows because conditional entropy is invariant under tensoring with another
state for the conditioning system, in turn being a special kind of conditionally unital channel. The inequality follows because the
swap and prepare channel is conditionally unital, as argued above. The
penultimate equality follows from additivity, and the final equality because
the entropy of the uniform state of dimension $\left\vert B\right\vert $ is
equal to $\log_{2}\left\vert B\right\vert $.
\end{proof}

\begin{proposition}
Let $\rho_{XA}$ be a classical--quantum state of the following form:%
\begin{equation}
\rho_{XA}=\sum_{x}p(x)|x\rangle\!\langle x|_{X}\otimes\rho_{A}^{x}.
\end{equation}
Then%
\begin{equation}
\mathbf{H}(X)_{\rho}\leq\mathbf{H}(AX)_{\rho}.
\end{equation}

\end{proposition}

\begin{proof}
For every state $\rho_{A}^{x}$, there is a unital channel $\mathcal{N}_{A}%
^{x}$ that prepares it from the pure state $|0\rangle\!\langle0|_{A}$. Then we
take the overall channel to be one that measures the classical system $X$ in
the standard basis $\{|x\rangle\}_{x}$ and conditioned on the outcome,
performs the channel $\mathcal{N}_{A}^{x}$ on the system $A$. Thus, the
overall channel is realized as%
\begin{equation}
\mathcal{N}_{XA}(\omega_{XA})=\sum_{x}|x\rangle\!\langle x|_{X}\otimes
\mathcal{N}_{A}^{x}(\langle x|_{X}\omega_{XA}|x\rangle_{X}).
\end{equation}
This channel is unital because
\begin{align}
& \mathcal{N}_{XA}(I_{XA}) \notag \\
& =\mathcal{N}_{XA}(I_{X}\otimes I_{A})\\
& =\sum_{x}|x\rangle\!\langle x|_{X}\otimes\mathcal{N}_{A}^{x}(\langle
x|_{X}\left(  I_{X}\otimes I_{A}\right)  |x\rangle_{X})\\
& =\sum_{x}|x\rangle\!\langle x|_{X}\otimes\mathcal{N}_{A}^{x}(I_{A})\\
& =\sum_{x}|x\rangle\!\langle x|_{X}\otimes I_{A}\\
& =I_{X}\otimes I_{A}.
\end{align}
Furthermore,%
\begin{equation}
\mathcal{N}_{XA}(\rho_{X}\otimes|0\rangle\!\langle0|_{A})=\rho_{XA}.
\end{equation}
This implies that%
\begin{align}
\mathbf{H}(X)_{\rho_{X}}  & =\mathbf{H}(XA)_{\rho_{X}\otimes|0\rangle
\langle0|_{A}}\\
& \leq\mathbf{H}(XA)_{\mathcal{N}_{XA}(\rho_{X}\otimes|0\rangle\!\langle0|_{A}%
)}\\
& =\mathbf{H}(XA)_{\rho_{XA}}.
\end{align}
This concludes the proof.
\end{proof}

\section{Example of a non-negative, weakly additive monotonic function}

\label{appg}

In this appendix, we prove that there exists a non-zero function that satisfies the following three properties:
\begin{enumerate}
    \item Monotonicity under conditional majorization; i.e., satisfying~\eqref{eq:QCE-1st-postulate-mono}.
    \item Weak additivity of the form~\eqref{weakadditivity}.
    \item Non-negative.
\end{enumerate}
The third property above demonstrates that the negativity of conditional entropy is closely related to the full additivity postulate in~\eqref{eq:QCE-2nd-postulate-add}. The construction of the example is based on a technique developed in~\cite{GT20} for the extension of resource monotones from one domain to another.

Let $\mathbb{H}$ be a classical conditional entropy. In~\cite{GT20} it was shown that every such function has two optimal extensions to the quantum domain. Here we are interested in the following extension. For every $\rho\in\mathfrak{D}(AB)$, we define
\be\label{g1}
f(A|B)_\rho\eqdef \inf_{\rho_{AB}\succ_A\sigma_{XY}}\mathbb{H}(X|Y)_{\sigma}\;,
\ee
where the infimum above is over all classical systems $X$ and $Y$ and every classical state $\sigma_{XY}$ with the property that $\rho_{AB}\succ_A\sigma_{XY}$.
The function $f$ has the following key properties~\cite{GT20}: 
\begin{enumerate}
    \item Reduction: For every classical state $\rho\in\md(XY)$ 
\be
f(X|Y)_\rho=\mathbb{H}(X|Y)_\rho\;.
\ee
\item Monotonicity: For all $\rho\in\md(AB)$ and $\omega\in\md(A'B')$ such that $\rho_{AB}\succ_A\omega_{A'B'}$,
\be
f(A'|B')_{\omega}\geq f(A|B)_\rho\;.
\ee
\item Subadditivity: The function $f$ is subadditive under tensor-product states.
\item Positivity: The function $f$ is non-negative.
\end{enumerate}
The last property follows from the definition in~\eqref{g1}. Since $f$ is subadditive, the limit
\be
g(A|B)_\rho\eqdef\lim_{n\to\infty}\frac1n f(A^n|B^n)_{\rho^{\otimes n}}
\ee
exists. Therefore, the function $g$ is weakly additive. Moreover, the function $g$ satisfies the first two properties above, namely, reduction to $\mathbb{H}$ on classical states and behaves monotonically under conditional majorization. Observe that the reduction of $g$ to $\mathbb{H}$ on classical states implies that it is not equal to the zero function since $\mathbb{H}$ is not the zero function. Since $g$ is non-negative it cannot be a conditional entropy. This means that $g$ cannot be fully additive, as this is the only property that does not follow directly from its definition.



\end{document}